\documentclass[notitlepage,openbib,12pt]{article}
\usepackage{amsfonts}
\usepackage{amsmath}
\usepackage{amssymb}
\usepackage{makeidx}
\usepackage{graphicx}
\usepackage{natbib}
\usepackage[margin=1in]{geometry}
\usepackage{color}
\usepackage{url}
\usepackage[pdftex,     colorlinks=true, linkcolor=blue, citecolor=blue, urlcolor=blue]{hyperref}
\usepackage{accents}

\newtheorem{theorem}{Theorem}

\newtheorem{claim}{Claim}

\newtheorem{definition}{Definition}

\newtheorem{lemma}{Lemma}

\newtheorem{remark}{Remark}

\newtheorem{observation}{Observation}

\newenvironment{proof}[1][Proof]{\noindent \textbf{#1.} }{\  \rule{0.5em}{0.5em}}
\input{tcilatex}
\begin{document}

\title{\textbf{Direct Implementation with Evidence\thanks{%
We are grateful to a coeditor and two anonymous referees for their
insightful comments which led to major improvements to the paper. We also
thank Bart Lipman, Hamid Sabourian, and Satoru Takahashi for helpful
comments and suggestions. Chen gratefully acknowledges the financial support
from the Singapore Ministry of Education Academic Research Fund Tier 1
(R-122-000-328-115). Sun gratefully acknowledges the financial support from
the National Natural Science Foundation of China (NSFC72273029).}}}
\author{Soumen Banerjee\thanks{%
China Center for Behavioural Economics and Finance, Southwestern University
of Finance and Economics. Email: soumen08@gmail.com} \and Yi-Chun Chen%
\thanks{%
Department of Economics and Risk Management Institute, National University
of Singapore. Email: ecsycc@nus.edu.sg} \and Yifei Sun\thanks{%
School of International Trade and Economics, University of International
Business and Economics. Email: sunyifei@uibe.edu.cn}}
\maketitle

\begin{abstract}
We study full implementation with evidence in an environment with bounded
utilities. We show that a social choice function is Nash implementable in a
direct revelation mechanism if and only if it satisfies the measurability
condition proposed by \cite{BL2012}. Building on a novel classification of
lies according to their refutability with evidence, the mechanism requires
only two agents, accounts for mixed-strategy equilibria and accommodates
evidentiary costs. While monetary transfers are used, they are off the
equilibrium and can be balanced with three or more agents. In a richer model
of evidence due to \cite{KT2012}, we establish pure-strategy implementation
with two or more agents in a direct revelation mechanism. We also obtain a
necessary and sufficient condition on the evidence structure for
renegotiation-proof bilateral contracts, based on the classification of lies.
\end{abstract}

\baselineskip=19.2pt

\section{Introduction}

Consider a government which aims to balance infrastructure development with
environmental protection. It consults infrastructure development firms and
environmental protection organizations to do a cost-benefit analysis.
Without knowing the exact state of the environment, the government needs to
factor in inputs from these consultants. However, the consultants have their
own incentives which are not necessarily aligned with that of the
government. For instance, infrastructure development firms will always
prefer to build rather than to not build, whereas environmental protection
organizations will always push to err on the side of caution in protecting
nature. How can the government glean useful information from these agents
whose incentives are conceivably misaligned?

Alternatively, consider a litigation scenario in which a firm is suing a
supplier for providing defective parts while the supplier argues that the
parts meet the specifications in the initial contract. The judge wants to
impose a financial penalty commensurate to the offence, assuming one is
proven. The plaintiff always prefers higher penalties while the defendant
always prefers smaller ones. How, then, is the judge to decide on the scale
of the penalty, given that those in the know do not have the incentive to
truthfully reveal whether the parts were defective or the contract unclear?

A common thread that links these scenarios is that the preferences of the
agents do not change across states. Other scenarios which share this feature
include budget allocation (where agents prioritize obtaining larger shares,
irrespective of their requirements), and lobbying (where groups prioritize
their own interests which are independent of the state). In all of these
situations, classical results on full implementation such as \cite{M99} and 
\cite{MR88} cannot be used, as they rely on preference variation across
states. In this paper, we pursue the idea of enriching the mechanism with
the use of evidence, so that agents can no longer misreport the state
arbitrarily. For instance, the infrastructure development firms or the
environmental protection organizations may be able to partially prove the
state of the environment by submitting their environmental research reports.
Likewise, the defendant in the litigation scenario may be able to prove that
the supplied product meets the specifications in the contract or the
plaintiff may be able to prove that it does not.

Specifically, we study the full implementation problem with evidence due to 
\cite{BL2012}. There is a state of the world which is common knowledge among
a set of agents but unknown to a designer. At each state, an agent is
endowed with some articles of evidence which may vary from one state to
another. Each article of evidence can be identified with a subset of the
state space, and it refutes the possibility that the true state is outside
this subset. In other words, we work with hard evidence which differs in its
availability across states and hence can be used by the agents to partially
prove the state.\footnote{%
The hard evidence setting is a special case of the costly evidence setting
due to \cite{KT2012} which we study in Section \ref{COSTLY}.} In designing a
mechanism, the planner can request evidence presentation as well as cheap
talk messages (which are available in every state). In such a setting, \cite%
{BL2012} propose a condition called \textit{measurability}. A social choice
function is said to be \textit{measurable} with respect to the underlying
evidence structure (hereafter, measurable) if whenever the desirable social
outcomes differ across two states, at least one agent has a variation in the
set of available evidence. To wit, if neither the preferences nor the
evidence varies between two states, then any (direct or indirect) mechanism
will have the same outcome in both, regardless of the solution concept to
which the planner subscribes.

We study an environment in which there are two or more agents with bounded
utilities and the designer can impose monetary transfers off the
equilibrium. We consider a \emph{normal} evidence structure under which each
agent can present (in one message) all the articles of evidence which he is
endowed with.\footnote{%
Normality is also imposed in Theorem 2 of \cite{BL2012} which achieves Nash
implementation by invoking integer games and $\varepsilon $ transfers. For a
detailed comparison between our results and the results of \cite{BL2012},
see Section \ref{literature}.} In this setting, we show that a social choice
function is implementable in mixed-strategy Nash equilibria regardless of
the agents' preferences if and only if it is measurable (Theorem \ref{main}%
). Moreover, we obtain the implementation result using a \emph{direct
revelation mechanism} in the sense of \cite{BW2007}, wherein agents only
report a state and present an article of evidence. Further, the mechanism
achieves budget balance when there are three or more agents; in addition, if
we allow the designer to randomize, then the off-the-equilibrium transfers
can be made arbitrarily small (Theorem \ref{small}).

Our implementation results possesses a number of desirable features. First,
we recognize, following the critique due to \cite{J1992}, that
implementation has long relied on invoking \emph{integer/modulo games} to
eliminate unwanted equilibria. While such devices are useful in achieving
positive results in general settings, they admit no pure-strategy
equilibrium. This problem is exacerbated under mixed strategies: in an
integer game, an agent has no best response to an opponent's strategy which
places positive probability on all integers, whereas a modulo game possesses
an unwanted mixed-strategy equilibrium. The hope has been that more
realistic mechanisms may suffice in more specific settings. Our mechanism
makes use of neither integer/modulo games nor sequential moves.\footnote{%
Implementation with equililbrium refinements which do not possess the
closed-graph property (e.g., subgame-perfect Nash equilibrium) need not be
robust to a \textquotedblleft small amount of incomplete information about
the state\textquotedblright ; see \cite{CE2003} and \cite{A2012}.} Rather,
we use a direct revelation mechanism which has the simplest possible message
structure in implementing arbitrary measurable social choice functions. The
design of our direct revelation mechanism is based on a novel classification
of lies (i.e., inaccurate state claims). Specifically, there are lies which
can be refuted by evidence possessed by other agents (other-refutable lies),
lies which can be refuted only by the agent who is reporting them
(self-refutable lies), and lies which cannot be refuted by any evidence
available under the true state (nonrefutable lies).\footnote{%
It follows from measurability that if state $s^{\prime }$ is not refutable
at state $s$ and induces a different social outcome, then $s$ must be
refutable at $s^{\prime }$.} We design transfer rules which eliminate the
lies successively in any mixed-strategy equilibrium.

Second, our result extends to settings where evidence presentation is
costly. Specifically, as long as the designer knows the bound of the agents'
evidentiary cost, every measurable social choice function remains directly
implementable in mixed-strategy Nash equilibria regardless of the cost
structure (Theorem \ref{hardcostly}).\footnote{%
It is well recognized in the literature that there may be material or
psychological costs associated with the presentation of evidence. See \cite%
{BW2007}, \cite{BL2012} and \cite{KT2012} for instance.} In contrast, when
there is a fixed evidentiary cost structure which is common knowledge, the
designer is able to distinguish between states by exploiting the cost
variation and measurability is no longer necessary for implementation. In
such a setting, \cite{KT2012} propose a condition called \textit{evidence
monotonicity}, and show that it is necessary for implementation where only
the cheapest evidence is submitted in equilibrium. We adapt the notion of
evidence monotonicity to our environment and establish that evidence
monotonicity is sufficient for pure-strategy Nash implementation with two or
more agents via a direct revelation mechanism (Theorem \ref{main_soft}).%
\footnote{%
We obtain mixed-strategy implementation under a stronger version of evidence
monotonicity in Theorem 6.}

Third, as our results hold even when there are only two agents, they allow
for applications to a classical issue in the incomplete contract literature.
The difficulty in this regard arises when a desirable contractual outcome
(e.g., an efficient trade) needs to be conditioned on some state variables
(e.g., the a buyer's valuation of some good) which are observable to both
contractual parties and yet not verifiable by a third party such as a court.
A well-known solution is to invoke implementation theory to design a
mechanism which has the agents announce the observed state as a verifiable
equilibrium message. However, such mechanisms often involve
off-the-equilibrium transfers which penalize both agents so that such
mechanisms are susceptible to renegotiation.\footnote{%
For an example of the issues posed by renegotiation, we refer the reader to 
\cite{MM99} where they demonstrate a setting in which only a null contract
is renegotiation-proof even though it is efficient to trade with verifiable
states. \cite{MT99} attempt to circumvent this issue by using lotteries, but
their result depends crucially on at least one agent being strictly risk
averse.} Measurability with respect to an evidence structure may be
considerably easier to satisfy in practice than verifiability. Indeed, even
when a state is not verifiable, the agents may still be able to provide
evidence to refute certain states. In Section \ref{RP}, we establish a
necessary and sufficient condition on the evidence structure for the
existence of renegotiation proof bilateral contracts. In particular, such a
contractual outcome must lie on the Pareto frontier of the agents' utility
possibility set and thereby must achieve budget balance.

The rest of the paper is organized as follows. Section \ref{Model} provides
a formal description of the model, the implementing condition, and further
details the classification of lies on which the mechanism is based. Section %
\ref{Mech} presents the main implementing mechanism and a formal proof of
implementation. Following this, we also establish budget balance with three
or more agents (Section \ref{BB}) and implementation with small transfers
(Section \ref{SmallTransfers}). In Section \ref{COSTLY}, we extend our
results to settings where evidence presentation is costly. Section \ref{RP}
details our treatment of renegotiation-proof contracting and we conclude by
comparing our results to the existing literature.

\section{\label{Model}Model}

\subsection{\label{Env}Environment}

Let $\mathcal{I=}\left\{ 1,...,I\right\} $ ($I\geq 2$) be a set of agents, $%
A $, a set of social outcomes, and $S$ a set of states. Suppose that $S$ is
finite. Agents have quasilinear utilities in transfers, so that, $%
u_{i}(a,s,\tau )=v_{i}(a,s)+\tau $ where $\tau $ is the transfer to the
agent. We assume that $v_{i}$ is bounded and without loss of generality, we
set $v_{i}:A\times S\rightarrow \left[ 0,1\right] $ (in dollars). As a
result, an agent can be induced to accept any outcome if the alternative
were to be any other outcome with a penalty of 1 dollar. A social planner
would like to implement a social choice function (SCF) $f:S\rightarrow A$.
We assume that while the true state is common knowledge among the agents, it
is unknown to the social planner.

\subsection{\label{HE}Evidence}

We assume that each agent $i$ is endowed with a (state-dependent) collection
of articles of evidence $\mathcal{E}_{i}:S\rightrightarrows 2^{S}$. In
particular, by providing a message $E_{i}\in \mathcal{E}_{i}\left( s\right) $
in a state $s$, agent $i$ establishes that the true state lies within $E_{i}$%
. At state $s$ we say that an article of evidence $E_{i}\in \mathcal{E}%
_{i}(s)$ \emph{refutes} state $s^{\prime }$ if $s^{\prime }\not\in E_{i}$.

Following \cite{BL2012}, we introduce the following definition:

\begin{definition}
\label{e1e2}An evidence structure satisfies the following conditions:\newline
(e1) it is impossible to refute the truth, i.e., $\forall s\in S,$ $E_{i}\in 
\mathcal{E}_{i}(s)$ only if $s\in E_{i}$; \newline
(e2) if an agent can prove the event $E$ in some state, then they must be
able to do so in all the states in $E$, i.e., $E\in \mathcal{E}_{i}(s)$ only
if $E\in \mathcal{E}_{i}(s^{\prime })$ for every $s^{\prime }\in E$.
\end{definition}

We stress here that this is not an assumption. Rather, if articles of
evidence differ in their availability across states, then without loss of
generality, we can \emph{name} the article in terms of the subset of states
in which it is available. With such names, the above properties must be
satisfied.

We say that a setting involves \textit{hard }evidence if the set of evidence
available to an agent can change from state to state. Furthermore, we say
that the evidence structure is \emph{normal} if, for every agent $i$, and
every state $s\text{{}}$,

\begin{equation*}
E_{i}^{\ast }\left( s\right) \equiv \bigcap_{E\in \mathcal{E}_{i}\left(
s\right) }E\in \mathcal{E}_{i}\left( s\right) \text{.}
\end{equation*}%
The idea behind normality is that it is feasible for agents to present all
their evidence at once, suggesting an idealization in which there are no
time or other constraints on doing so. We refer to this message containing
all the evidence that an agent has as the \emph{tightest} evidence of the
agent in the given state. For a start, we prove our first main result
(Theorem \ref{main}) under the assumption of normality. We will discuss in
Section \ref{Norm} how this result should be modified when the normality
assumption does not hold. We define a social choice environment as a tuple $%
\Psi =\left( \mathcal{I},A,S,\{\mathcal{E}_{i}\}_{i\in \mathcal{I}},f\right) 
$ which is assumed to be common knowledge among the designer and agents. We
write $s\sim s^{\prime }$ (and say that $s$ is equivalent to $s^{\prime }$)
if $\mathcal{E}_{i}\left( s\right) =\mathcal{E}_{i}\left( s^{\prime }\right) 
$ for all $i$.

\subsection{\label{Example}Illustrating Example}

To illustrate the above ideas, we consider a situation wherein a government
(acting as the social planner) is considering whether to approve a
development project which has an adverse environmental impact. The agents
concerned are an infrastructure development firm (F) and an environmental
protection organization (O). Suppose the three actions available to the
planner are to allow the project ($P_{l}$), ask for it to be made smaller ($%
P_{m}$), or to scrap it entirely ($P_{h}$). It wishes to make these
decisions if the threat to the environment is low ($s^{l}$), medium ($s^{m}$%
) or high ($s^{h}$) respectively. The government does not know the degree of
the threat to the environment, while both F and O, being more familiar with
the scenario, know that the true threat level is medium ($s^{m}$).

Profit-making firm (F) is solely concerned with implementing the project and
thus has a preference ordering $P_{l}\succ P_{m}\succ P_{h}$. The
environmental protection organization (O) has the preference ordering $%
P_{h}\succ P_{m}\succ P_{l}$ since it strictly prioritizes protecting the
environment. These preference orderings are not dependent on the state and
this is where this scenario departs from the classical/Maskin-type
implementation problem.

The social planner now requests the two parties to submit evidence for and
against the project. Suppose that articles of evidence are of the form of
discoveries of environmental degradation which refute lower-impact states.
Moreover, the state $s^{h}$ is assumed to be associated with a very serious
degree of environmental impact, so that if the true state were $s^{h}$, it
would be apparent to (and provable by) all the agents. To sum up these
ideas, we present below the evidence structure for each agent in all
possible states:%
\begin{equation*}
\begin{tabular}{|c|c|c|c|}
\hline
Agent/State & $s^{l}$ & $s^{m}$ & $s^{h}$ \\ \hline
Firm (F) & $\{\{s^{l},s^{m},s^{h}\}\}$ & $\{\{s^{l},s^{m},s^{h}\},%
\{s^{m},s^{h}\}\}$ & $\{\{s^{l},s^{m},s^{h}\},\{s^{m},s^{h}\},\{s^{h}\}\}$
\\ \hline
Organization (O) & $\{\{s^{l},s^{m},s^{h}\}\}$ & $\{\{s^{l},s^{m},s^{h}\}\}$
& $\{\{s^{l},s^{m},s^{h}\},\{s^{h}\}\}$ \\ \hline
\end{tabular}%
\end{equation*}

Note that articles of evidence are subsets of the state space so that an
agent presenting $\{s^{m},s^{h}\}$ informs the designer only that the state
is not $s^{l}$. We also say that this article of evidence refutes $s^{l}$.
Thus, in state $s^{m}$, reporting the state $s^{l}$ is a lie that only F can
refute, and reporting the state $s^{h}$ is a lie that no one can refute. It
is interesting to note that in the true state ($s^{m}$), O would like to
convince the planner that the true state is $s^{h}$ (a claim that no one can
refute in this scenario) and F would like to convince the planner that the
true state is $s^{l}$ (a claim that only F can refute).

\subsection{Measurability and implementation}

When the agents' preferences do not vary across states (i.e., $v$ does not
depend on $s$), evidence is the only way to differentiate two states.
Indeed, if two states induce the same preference profile and evidence
endowments, then irrespective of the solution concept in use, in any
mechanism they must be associated with the same set of equilibria.\footnote{%
Unless otherwise specified, we assume that articles of evidence do not have
costs associated with them so that the designer cannot exploit cost
variation.} Measurability of an SCF entails that when the planner wants to
implement different outcomes from one state to another, there must be at
least one agent who can differentiate between the states in terms of their
evidence set. We now state the formal definition.

\begin{definition}
Given a social environment $\Psi $, an SCF $f$ satisfies measurability if $%
f\left( s\right) =f\left( s^{\prime }\right) $ whenever $s\sim s^{\prime }$.
\end{definition}

In other words, if $f$ is measurable, then in implementing the desirable
social outcome, the designer needs to identify only the equivalence class of
states which contains the truth.

To fix ideas, consider the example in Section \ref{Example} with the
modification that in state $s^{m}$, F is endowed only with $%
\{\{s^{l},s^{m},s^{h}\}\}$. In this case, no agent can differentiate between 
$s^{l}$ and $s^{m}$, so that any implementable SCF must be constant between
these two states. More generally, if there were to be no evidence in the
model, so that $\mathcal{E}_{i}\left( s\right) =\{S\}$ for all agents in all
states, then only constant SCFs are implementable. Further, even if each
agent were endowed only with $\{\{s^{l},s^{m},s^{h}\}\}$ in each state with
the exception of only one agent being endowed with $\{s^{h}\}$ in state $%
s^{h}$, it (we shall show later) would still be possible to implement
different outcomes in $s^{h}$ relative to $s^{l}$ and $s^{m}$, so that even
the slightest variation in evidence is sufficient for implementing different
outcomes.

A mechanism $\mathcal{M}$ in this social choice environment is defined as a
tuple $\mathcal{M}=(M,g,\left( \tau _{i}\right) _{i\in \mathcal{I}})$ where $%
M=\Pi _{i\in \mathcal{I}}M_{i}$ is a finite set of message profiles, $%
g:M\rightarrow A$ is the outcome function, and $\tau _{i}:M\rightarrow 
\mathbb{%
\mathbb{R}
}$ is the payment rule for agent $i$. A mechanism $\mathcal{M}$ together
with a profile of utility functions $v=\left( v_{i}\right) _{i\in \mathcal{I}%
}$ with $v_{i}:A\times S\rightarrow \left[ 0,1\right] $ induces a
complete-information game $G\left( \mathcal{M},v,s\right) $ at state $s$. We
call a mechanism a \emph{direct revelation mechanism} (\cite{BW2007}) when $%
M_{i}=S\times \mathcal{E}_{i}$, that is, every agent submits only one claim
of state and one article of evidence.

A (mixed) \emph{strategy} of agent $i$ in the game $G\left( \mathcal{M}%
,v,s\right) $ is a probability distribution $\sigma _{i}$ over $M_{i}$,
which we also denote by $\sigma _{i}\in \Delta M_{i}$. A strategy profile $%
\sigma =(\sigma _{1},...,\sigma _{I})\in \times _{i\in \mathcal{I}}\Delta
M_{i}$ is said to be a (mixed-strategy) \emph{Nash equilibrium} of the game $%
G\left( \mathcal{M},v,s\right) $ if, for any agent $i\in \mathcal{I}$ and
for any messages $m_{i}\in $supp$(\sigma _{i})$, we have 
\begin{eqnarray*}
&&\sum_{m_{-i}\in M_{-i}}\sigma _{-i}(m_{-i})[v_{i}(g(m_{i},m_{-i}),s)+\tau
_{i}(m_{i},m_{-i})] \\
&\geq &\sum_{m_{-i}\in M_{-i}}\sigma _{-i}(m_{-i})[v_{i}(g(m_{i}^{\prime
},m_{-i}),s)+\tau _{i}(m_{i}^{\prime },m_{-i})],\forall m_{i}^{\prime }\in
M_{i}\text{.}
\end{eqnarray*}

\noindent where $\sigma _{-i}(m_{-i})=\Pi _{j\neq i}\sigma _{j}\left(
m_{j}\right) $. A pure-strategy Nash equilibrium is a Nash equilibrium $%
\sigma $ which assigns probability one to some message profile $m$.

Unlike the classical implementation problem, and in recognition of practical
situations in which preferences do not vary across states, we seek to
implement the SCF by relying on evidence instead of preference reversal. To
stress this difference, we require in the following definition that
implementation obtain regardless of the profile of utility functions. This
also has the effect of strengthening our result from the point of view that
we no longer require that preferences vary between states for implementation.%
\footnote{%
Evidence variation however, continues to be necessary. \cite{KT2012} provide
the minimum necessary condition in this context, when both preference and
evidence variation are combined - evidence monotonicity. For the purposes of
this paper, we set aside preference variation and focus on evidence
variation alone.}

\begin{definition}
\label{implementation}An SCF $f$ is Nash-implementable if there is a
mechanism $\mathcal{M}=(M,g,\left( \tau _{i}\right) _{i\in \mathcal{I}})$
such that for any profile of bounded utility functions $v=\left(
v_{i}\right) _{i\in \mathcal{I}}$, any state $s$, and any mixed-strategy
Nash equilibrium $\sigma $ of the game $G\left( \mathcal{M},v,s\right) $, $%
g(m)=f(s)$ and $\tau _{i}(m)=0$ for each message profile $m\in $supp$\sigma
\left( s\right) $.
\end{definition}

That is, a mechanism implements a social choice function if at any state,
and with any profile of bounded utility functions, any mixed-strategy Nash
equilibrium outcome coincides with the outcome of the SCF. That is, we ask
for full implementation as opposed to partial implementation that requires
only one equilibrium achieving the outcome of the SCF. Note that partial
implementation is trivial in such a setting, as all that would be required
to achieve it would be to heavily penalize all agents for disagreeing with
each other about the state, so that there would be an equilibrium where each
agent tells the truth and there are no transfers.

In what follows, we consider two states to be different only if they induce
different tightest evidence for at least one agent. Otherwise, they are
treated as the same state. In other words, we identify $S$ with its quotient
space $S/\sim $ induced by the equivalence relation $\sim $ where each point
corresponds to an equivalent class. Owing to the necessity of measurability,
this is without loss of generality for our implementation exercise.

\subsection{A Classification of Lies}

We will construct a direct revelation mechanism which leverages two ways to
use evidence in cross-checking any claim of state. First, an article of
evidence may be able to \emph{refute} a state claim. That is, it establishes
that the state claim is definitely not true. Second, it is possible for the
designer to pick out state claims which have not been fully \emph{supported}%
\textbf{\ }by agents, that is, states for which agents have not provided all
the evidence they ought to have if the state were true. If a state claim is
not supported by an agent even though he is incentivized to do so, then it
signals to the designer that the state claim is false. These two ideas
underlie the mechanism which we will present later.

Formally, agent $i$ is said to have \emph{supported} a state claim $s$ in an
evidence message $E_{i}$ if $E_{i}\subseteq E_{i}^{\ast }(s)$. Based upon
the notion of refutation, we distinguish three different types of lies
wherein a lie is a claim of state $s^{\prime }$ which is different from $%
s^{\ast }$ (under the relation $\sim $ defined previously). First, an \emph{%
other-refutable} lie for agent $i$ is a lie that at least one agent other
than $i$ has the evidence to refute in the true state. For instance, for
organization O, the lie $s^{l}$ is an other-refutable lie in state $s^{m}$
since it can be refuted by the article $\{s^{m},s^{h}\}$ possessed by Firm
F. Note that it is straightforward to construct a transfer rule so that no
agent will tell other-refutable lies; this is done by just requiring an
agent to pay a large penalty to whoever refutes his state claim.

Second, a \emph{self-refutable} lie for agent $i$ is a lie that only agent $%
i $ has the evidence to refute in the true state. For instance, for firm F,
the lie $s^{l}$ is a self-refutable lie in state $s^{m}$ since it can only
be refuted by the article $\{s^{m},s^{h}\}$ possessed by the firm itself.

Finally, a \emph{nonrefutable lie} is a lie that cannot be refuted by any
evidence that is possessed by any agent. For instance, the lie $s^{h}$ is a
nonrefutable lie in state $s^{m}$.

From an agent's perspective, the truth, other-refutable, self-refutable and
non-refutable lies partition the entire state space. To see this, notice
that given a lie, it can either be refuted at the true state, or not. If it
can be refuted, it can either be refuted by other agents, or only by the
agent in question.

We now prove the following observations which will be exploited in proving
our main result.

\begin{observation}
\label{c_NRL_0}If an agent $i$ cannot refute $s^{\prime }$ at $s^{\ast }$,
then every article of evidence available to him at $s^{\ast }$ is also
available to him at $s^{\prime }$.
\end{observation}

\begin{proof}
If $s^{\prime }\notin RL_{i}\left( s^{\ast }\right) $, then every article of
evidence available to $i$ at $s^{\ast }$ contains $s^{\prime }$. Then, from
Property (e2) of Definition \ref{e1e2}, every such article is available to $%
i $ at $s^{\prime }$, so that $\mathcal{E}_{i}\left( s^{\ast }\right)
\subseteq \mathcal{E}_{i}\left( s^{\prime }\right) $.
\end{proof}

\begin{observation}
\label{c_SRL_0}If $s^{\prime }$ is a nonrefutable lie at $s^{\ast }$, then
some agent must have an article of evidence at state $s^{\prime }$ which
refutes $s^{\ast }$.
\end{observation}

\begin{proof}
As $s^{\prime }$ is nonrefutable for $i$ at $s^{\ast }$, Observation \ref%
{c_NRL_0} yields that $\mathcal{E}_{i}\left( s^{\ast }\right) \subseteq 
\mathcal{E}_{i}\left( s^{\prime }\right) $ for every $i$. Since $s^{\prime }$
is a lie, $s^{\prime }\nsim s^{\ast }$. Hence, $\mathcal{E}_{i}\left(
s^{\ast }\right) \subset \mathcal{E}_{i}\left( s^{\prime }\right) $ for some 
$i$, and from (e2) of Definition \ref{e1e2} any member of $\mathcal{E}%
_{i}\left( s^{\prime }\right) \backslash \mathcal{E}_{i}\left( s^{\ast
}\right) $ must refute $s^{\ast }$. In other words, any article of evidence
that is available to agent $i$ under $s^{\prime }$ and not available under $%
s^{\ast }$ must refute $s^{\ast }$.
\end{proof}

Observation \ref{c_NRL_0} establishes that supporting either a
self-refutable lie of another agent, or a nonrefutable lie requires the
presentation of at least as much evidence as supporting the truth.
Observation \ref{c_SRL_0} entails that it is impossible to have every agent
support a nonrefutable lie. Since nonrefutable lies cannot be directly
refuted by evidence, this inability to support them forms the only way to
eliminate them in equilibrium.

\section{\label{Mech}Implementing Mechanism}

Fix the environment $\Psi =(\mathcal{I},A,S,\{\mathcal{E}_{i}\}_{i\in 
\mathcal{I}},f)$. We now present our main result.

\begin{theorem}
\label{main}Suppose the evidence structure is given by $\mathcal{E}%
_{i}\left( \cdot \right) $. Then, an SCF $f$ is Nash-implementable in a
direct revelation mechanism if and only if it is measurable with respect to $%
\mathcal{E}_{i}\left( \cdot \right) $.
\end{theorem}

The necessity of measurability follows from the fact that if two states are
associated with the same set of evidence, and the preferences are constant
among states, then any mechanism must have the same set of equilibria in
both states. In the following subsections, we prove the sufficiency part of
Theorem \ref{main} by constructing a direct revelation mechanism which
implements $f$.

\subsection{\textbf{Message Space}}

Every agent has a typical message $m_{i}=(s_{i},E_{i})\in M_{i}=S\times 
\mathcal{E}_{i}$. We interpret this as a claim of state and an article of
evidence. The typical message $m_{i}$, therefore, is of the form $%
(s_{i},E_{i})$. In the following, we denote the full message profile (of all
agents) by $m\in \Pi _{i\in \mathcal{I}}M_{i}$.

\subsection{\textbf{Outcome}}

We define the outcome function of the mechanism as%
\begin{equation*}
g\left( m\right) =f\left( s_{1}\right) \text{.}
\end{equation*}%
That is, we implement the social outcome according to the state claim made
by the first agent. However, we will show that in any equilibrium all agents
report the true state. Hence, any Nash equilibrium achieves the desirable
social outcome.

\subsection{\textbf{Transfers}}

There are four different types of transfers in the mechanism, which we will
introduce one by one. The first transfer applies when an agent refutes a
state claim of another agent. In this case, the agent whose claim is refuted
has to pay a penalty to the agent who refutes the claim. That is,%
\begin{equation*}
\tau _{ij}^{1}\left( m\right) =\left\{ 
\begin{tabular}{ll}
$2I+1$, & if $s_{i}\in E_{j}$ and $s_{j}\not\in E_{i}$; \\ 
$-2I-1$, & if $s_{i}\not\in E_{j}$ and $s_{j}\in E_{i}$; \\ 
$0$, & otherwise.%
\end{tabular}%
\right.
\end{equation*}%
where $I$ is the number of agents.

Under the second transfer, an agent incurs a penalty if his state claim is
not supported by himself or other agents. Formally,

\begin{equation*}
\tau _{i}^{2}\left( m\right) =\left\{ 
\begin{tabular}{ll}
$-I$, & if $\exists j\in \mathcal{I}$ s.t. $E_{j}\not\subseteq E_{j}^{\ast
}(s_{i})$; \\ 
$0\,$, & otherwise.%
\end{tabular}%
\right.
\end{equation*}

The third transfer penalizes an agent (say, agent $i$) if he disagrees with
another agent (say, agent $j$) along the evidence dimension for agent $i$.
This is expressed as follows:

\begin{equation*}
\tau _{ij}^{3}\left( m\right) =\left\{ 
\begin{tabular}{ll}
$-1$, & if $E_{i}^{\ast }(s_{i})\not=E_{i}^{\ast }(s_{j})$; \\ 
$0$, & $\text{otherwise.}$%
\end{tabular}%
\right.
\end{equation*}

The fourth transfer is a penalty proportional to the cardinality of states
that are not refuted by the evidence presented by agent $i$. This is active
when an agent has made a state claim which one or more agents have not
supported. Formally,

\begin{equation*}
\tau _{i}^{4}\left( m\right) =\left\{ 
\begin{tabular}{ll}
$-\frac{|E_{i}|}{|S|}$, & if $E_{j}\not\subseteq E_{j}^{\ast }(s_{j^{\prime
}})$ for some $j,j^{\prime }\in \mathcal{I}$; \\ 
$0$, & otherwise.%
\end{tabular}%
\right.
\end{equation*}

We stress that this applies to one's own state claims as well, that is,
being unable to provide the tightest evidence for one's own state claim also
incurs a penalty from this transfer.

With $\tau _{i}^{1}=\sum_{j\neq i}\tau _{ij}^{1}$ and $\tau
_{i}^{3}=\sum_{j\neq i}\tau _{ij}^{3}$, we define the overall transfer to
agent $i$ as:

\begin{equation*}
\tau ^{i}=\tau _{i}^{1}+\tau _{i}^{2}+\tau _{i}^{3}+\tau _{i}^{4}\text{.}
\end{equation*}

\subsection{\textbf{Proof Sketch }}

The mechanism deals with each type of lie in sequence. We begin with
other-refutable lies, for instance the lie $s^{l}$ for O in state $s^{m}$ in
the illustrating example. This involves the transfer $\tau ^{1}$. Recall
that by definition, an other-refutable lie for some player (say $i$) is
refutable by a different agent, (say $j$). We note that $\tau ^{1}$ provides
agent $j$ the incentive to refute agent $i$'s lie irrespective of the
probability with which $i$ presents it (normality ensures that he does not
have to sacrifice rewards from refuting another agent's lies). It also
assures agent $i$ of a large penalty from refutation, which can be avoided
by deviating to the truth (which is irrefutable).

Before eliminating the remaining lies, we comment on the role of $\tau ^{4}$%
. The transfer $\tau ^{4}$ is a cardinality transfer which (when it is
active) incentivizes the presentation of the tightest evidence by all
agents. We construct the mechanism in a way that presenting additional
evidence is never harmful for an agent. Indeed, when $\tau ^{4}$ is active,
it strictly benefits the agent, since presenting an additional article of
evidence reduces $|E_{i}|$ and thus reduces the magnitude of the fine. This
allows the mechanism to elicit the true profile of tightest evidence, which
is useful for the elimination of both nonrefutable and self-refutable lies.
This crucially depends upon normality, since incentivizing agents to present
their tightest evidence would not achieve its goal unless such an article is 
\emph{available} and \emph{can be} presented.

Now, we eliminate nonrefutable lies, for instance the lie $s^{h}$ in state $%
s^{m}$ in the illustrating example. This step involves $\tau ^{2}$ and $\tau
^{4}$. Recall that if $s^{\prime }$ is a nonrefutable lie at $s^{\ast }$,
then some agent must have the evidence to refute $s^{\ast }$ at $s^{\prime }$%
; for instance, $\left\{ s^{h}\right\} $ refutes $s^{m}$ at $s^{h}$. Since
by assumption it is not possible to eliminate $s^{\ast }$ if it is the
truth, there is at least one agent who cannot support the claim $s^{\prime }$
if the true state is $s^{\ast }$. Therefore, if an agent $i$ presents a
nonrefutable lie, then he knows that it cannot be supported. In the
mechanism, this has two effects. First, he gets a large fine from $\tau ^{2}$%
. Second, $\tau ^{4}$ is active, so that all agents have a strict incentive
to present all their evidence. This has the consequence of allowing agent $i$
to evade the fine from $\tau ^{2}$ by switching to the truth (if everyone
presents all their evidence, the truth is supported by everyone).

If we prioritize the first two steps, then agents are restricted to
presenting either the truth or self-refutable lies, each of which (from
Observation \ref{c_NRL_0}) induces for any other agent an evidence set at
least as large as the truth. In this case, we claim that all agents present
their tightest evidence. Indeed, if any agent withholds evidence, they fail
to support the state claims of all other agents, so that $\tau ^{4}$ is
active, and submitting additional evidence is a profitable deviation. This
step, which involves $\tau ^{2}$, $\tau ^{3}$ and $\tau ^{4}$, uses this
fact to deal with self-refutable lies.

Finally, suppose that an agent $i$ presents a self-refutable lie, for
instance F claims $s^{l}$ in state $s^{m}$. Since all evidence is tightest,
the transfer $\tau ^{2}$ incentivizes other agents to present state claims
that are consistent with agent $i$'s tightest evidence. Following this, the
cross-check in $\tau ^{3}$ yields a penalty for agent $i$ because a
self-refutable lie for agent $i$ is inconsistent with agent $i$'s own
tightest evidence. This penalty can be avoided by switching to the truth,
since all agents are presenting their tightest evidence. In summary, in any
equilibrium everyone presents the truth with the tightest evidence. Hence,
implementation obtains.

\subsection{Proof of Implementation}

We present below a formal proof of implementation using the direct
revelation mechanism in Section \ref{Mech}. In what follows, we denote the
true state by $s^{\ast }$. For simplicity we will write $NRL$ and $SRL_{i}$
to denote $NRL\left( s^{\ast }\right) $ and $SRL_{i}\left( s^{\ast }\right) $%
, respectively. Fix an arbitrary mixed-strategy Nash equilibrium $\sigma $.

\subsubsection{Preliminary Results}

We begin by proving a lemma which finds use at multiple points in the proof.

\begin{lemma}
\label{l_tight}Given the strategies of other agents, an agent never incurs a
loss from presenting more evidence (while holding the state claim fixed) in
a deviation. Moreover, if $\tau _{i}^{4}$ is active with non-zero
probability, then under any optimal strategy, agent $i$ presents the
tightest evidence.
\end{lemma}

\begin{proof}
Fix agent $i$. It is clear that $\tau _{i}^{1}$ causes no loss from
presenting additional evidence. Indeed if agent $i$ refutes another agent's
state claim, he may make a profit from tightening the evidence. $\tau
_{i}^{2}$ only requires that the evidence presented be as tight as some
bound, so that tightening the evidence causes no loss from $\tau _{i}^{2}$
either. $\tau _{i}^{3}$ is a statement on state claims, so evidence is not
related to this transfer. $\tau _{i}^{4}$ clearly causes no loss from
tightening, rather, it causes strict gains if it is active. Therefore, given
a strategy profile for other agents, if an agent expects $\tau ^{4}$ to be
active with positive probability, then it is optimal for him to present the
tightest evidence, as otherwise he could tighten his evidence to improve his
payoff.
\end{proof}

\subsubsection{Eliminating Other-Refutable Lies}

\begin{claim}
\label{c_ORL_1}If agent $i$ reports with positive probability a message $%
m_{i}=(s_{i},E_{i})$ such that agent $j\neq i$ has an article of evidence
which refutes $s_{i}$, then $E_{j}$ must refute $s_{i}$ for every message $%
m_{j}=\left( s_{j},E_{j}\right) $ which agent $j$ reports with positive
probability.
\end{claim}

\begin{proof}
Suppose not. Then, agent $j$'s evidence does not refute $s_{i}$, that is, $%
s_{i}\in E_{j}$. Consider an alternate message $\tilde{m}_{j}$ which only
replaces $E_{j}$ with $E_{j}^{\ast }\left( s^{\ast }\right) $, which refutes 
$s_{i}$. The following table documents agent $j$'s payoff change by
deviating from $m_{j}$ to $\tilde{m}_{j}$:%
\begin{equation*}
\begin{tabular}{|l|l|l|l|l|l|}
\hline
$g$ & $\tau _{i}^{1}$ & $\tau _{i}^{2}$ & $\tau _{i}^{3}$ & $\tau _{i}^{4}$
& In total \\ \hline
$0$ & $>0$ & $0$ & $\geq 0$ & $\geq 0$ & $>0$ \\ \hline
\end{tabular}%
\end{equation*}

\noindent That is, agent $j$ gains from $\tau ^{1}$, on account of the fact
that he has now refuted agent $i$'s lie. By Lemma \ref{l_tight}, there is no
loss from any of the other transfers.
\end{proof}

\begin{claim}
\label{c_ORL_2}No agent reports an other-refutable lie with positive
probability.
\end{claim}

\begin{proof}
From Claim \ref{c_ORL_1}, if agent $i$ reports a lie $s_{i}$ which agent $j$
can refute, then with probability one agent $j$ must present $E_{j}$ to
refute $s_{i}$. Consider an alternative message $\tilde{m}_{i}$ which
replaces $s_{i}$ with $s^{\ast }$ in $m_{i}$. The following table summarizes
the payoff change by deviating from $m_{i}$ to $\tilde{m}_{i}$:%
\begin{equation*}
\begin{tabular}{|l|l|l|l|l|l|}
\hline
$g$ & $\tau _{i}^{1}$ & $\tau _{i}^{2}$ & $\tau _{i}^{3}$ & $\tau _{i}^{4}$
& In total \\ \hline
$>-1$ & $2I+1$ & $\geq -I$ & $\geq -I+1$ & $\geq -1$ & $>0$ \\ \hline
\end{tabular}%
\end{equation*}

Agent $i$ gains a minimum of $2I+1$ from $\tau ^{1}$ due to the fact that
agent $j$ was refuting $s_{i}$. In addition, in the worst case, agent $i$
loses at most $1$ from outcome $g$, at most $I$ from $\tau ^{2}$, at most $%
I-1$ from $\tau ^{3}$, and less than $\frac{|E_{i}|}{|S|}$ (which is less
than $1$) from $\tau ^{4}$. Hence, $\tilde{m}_{i}$ is a profitable deviation
from $m_{i}$.
\end{proof}

\subsubsection{Eliminating Nonrefutable Lies}

The next two claims eliminate the possibility that nonrefutable lies are
reported in equilibrium.

\begin{claim}
\label{c_NRL_2}If an agent reports with positive probability a message $%
m_{i}=(s_{i},E_{i})$ where $s_{i}$ is a nonrefutable lie at $s^{\ast }$,
then $E_{i}$ is the tightest, and every agent $j\neq i$ must provide the
tightest evidence available, i.e., $E_{j}=E_{j}^{\ast }(s^{\ast })$.
\end{claim}

\begin{proof}
Since $s_{i}\in NRL$, it follows from Observation \ref{c_SRL_0} that there
is an agent $j\in \mathcal{I}$ who can refute $s^{\ast }$ at $s_{i}$.
However, since the evidence structure $\mathcal{E}_{i}\left( \cdot \right) $
satisfies condition (e1) of Definition \ref{e1e2}, agent \thinspace $j$
cannot present $E_{j}^{\ast }(s_{i})$ (which must refute $s^{\ast }$) at the
state $s^{\ast }$. Thus, whenever there is an agent $i$ who reports with
positive probability a message $m_{i}$ with a nonrefutable lie $s_{i}$, then
for every $j\neq i$, $\tau _{j}^{4}$ must be triggered with positive
probability. It then follows from Lemma \ref{l_tight} that each agent $j\neq
i$ must present $E_{j}=E_{j}^{\ast }(s^{\ast })$ under any optimal strategy.
Further, in presenting $s_{i}$, agent $i$ knows that he (or another agent)
will be unable to support $s_{i}$ so that $\tau _{i}^{4}$ is active with
probability 1. Therefore, Lemma \ref{l_tight} yields $E_{i}=E_{i}^{\ast
}(s^{\ast })$ since agent $i$ plays an optimal strategy.
\end{proof}

\begin{claim}
\label{c_NRL_3}No agent reports a nonrefutable lie with positive probability.
\end{claim}

\begin{proof}
Suppose not. By Claim \ref{c_NRL_2}, if there is an agent who reports a
nonrefutable lie $s_{i}$ in message $m_{i}=(s_{i},E_{i})$, then $%
E_{i}=E_{i}^{\ast }(s^{\ast })$ and every agent $j\neq i$ will present $%
E_{j}=E_{j}^{\ast }(s^{\ast })$. Now, consider an alternative message $%
\tilde{m}_{i}$ which replaces $s_{i}$ with $s^{\ast }$. The following table
summarizes the payoff change by deviating from $m_{i}$ to $\tilde{m}_{i}$:%
\begin{equation*}
\begin{tabular}{|l|l|l|l|l|l|}
\hline
$g$ & $\tau _{i}^{1}$ & $\tau _{i}^{2}$ & $\tau _{i}^{3}$ & $\tau _{i}^{4}$
& In total \\ \hline
$>-1$ & $0$ & $\geq I$ & $\geq -I+1$ & $\geq 0$ & $>0$ \\ \hline
\end{tabular}%
\end{equation*}%
First, fixing $m_{i}$ where $s_{i}$ is nonrefutable$,$ we know that $\tau
_{i}^{2}(m_{i},m_{-i})=-I$ for every $m$ since there exists some agent $j$ ($%
j$ may be $i$) such that $E_{j}^{\ast }(s_{i})\subset E_{j}^{\ast }(s^{\ast
})$. Hence, agent $i$ gains at least $I$ from $\tau _{i}^{2}$ from the
deviation. Second, he loses less than $1$ from changing the outcome, incurs
no loss from $\tau _{i}^{1}$ (since the truth is not refutable), and at most 
$I-1$ from $\tau _{i}^{3}$. Third, agent $i$ incurs no losses from $\tau
_{i}^{4}$ either as he was facing a scenario of no support with $s_{i}$
before the deviation and the size of his evidence set has not changed.
Overall, this is, therefore, a profitable deviation.
\end{proof}

\subsubsection{Eliminating Self-Refutable Lies}

The next three claims eliminate self-refutable lies. Up to this point, we
have established that agents can report only the truth or self-refutable
lies in equilibrium. Observation \ref{c_NRL_0} implies that in this
situation, agents would need to present their tightest evidence to support
other agent's state claims. This is the basic idea which we will exploit in
establishing the following claims.

\begin{claim}
\label{c_SRL_2}All agents present the tightest evidence with probability
one. That is, for any agent $i$, $E_{i}=E_{i}^{\ast }\left( s^{\ast }\right) 
$ for any message $m_{i}=(s_{i},E_{i})$ on the support of $\sigma _{i}$.
\end{claim}

\begin{proof}
Suppose to the contrary that for some agent $i$, $m_{i}=(s_{i},E_{i})$ is on
the support of $\sigma _{i}$, and $E_{i}\not=E_{i}^{\ast }\left( s^{\ast
}\right) .$ From Claims \ref{c_ORL_2} and \ref{c_NRL_3}, other agents report
either a self-refutable lie or the truth in equilibrium. Therefore, from
Observation \ref{c_NRL_0}, if $E_{i}\not=E_{i}^{\ast }\left( s^{\ast
}\right) $, then agent $i$ expects $\tau ^{4}$ to be active with positive
probability as he is not supporting any other agent's claims. Deviating to $%
\tilde{m}_{i}=(s_{i},E_{i}^{\ast }\left( s^{\ast }\right) )$ from $m_{i}$ is
a profitable deviation in this case.
\end{proof}

\begin{claim}
\label{c_SRL_3}For every agent $i$, and for any message $m_{i}=\left(
s_{i},E_{i}\right) $ on the support of $\sigma _{i}$, we must have $%
E_{j}^{\ast }(s_{i})=E_{j}^{\ast }(s^{\ast })$ for every agent $j\neq i$.
\end{claim}

\begin{proof}
Suppose to the contrary that agent $i$ reports $m_{i}=\left(
s_{i},E_{i}\right) $ with $E_{j}^{\ast }(s_{i})\not=E_{j}^{\ast }(s^{\ast })$
for some agent $j\neq i$. It follows from Claims \ref{c_ORL_2} and \ref%
{c_NRL_3} that $s_{i}$ is a self-refutable lie for agent $i$. Since $%
E_{j}^{\ast }(s_{i})\not=E_{j}^{\ast }(s^{\ast })$, Observation \ref{c_NRL_0}
implies that $E_{j}^{\ast }(s_{i})\subset E_{j}^{\ast }(s^{\ast })$. From
Claim \ref{c_SRL_2}, all agents present the tightest evidence with
probability one, that is, $E_{j}=E_{j}^{\ast }(s^{\ast })$ with $\sigma _{j}$%
-probability one for every agent $j\in \mathcal{I}$.

Consider a deviation for agent $i$ to the truth, i.e., consider a message $%
\tilde{m}_{i}$ which only replaces $s_{i}$ with $s^{\ast }$ in $m_{i}$. The
following table summarizes the payoff change by deviating from $m_{i}$ to $%
\tilde{m}_{i}$:%
\begin{equation*}
\begin{tabular}{|l|l|l|l|l|l|}
\hline
$g$ & $\tau _{i}^{1}$ & $\tau _{i}^{2}$ & $\tau _{i}^{3}$ & $\tau _{i}^{4}$
& In total \\ \hline
$>-1$ & $\geq 0$ & $\geq I$ & $\geq -I+1$ & $\geq 0$ & $>0$ \\ \hline
\end{tabular}%
\end{equation*}

The agent gains $I$ from $\tau ^{2}$. This is because, by assumption, $%
E_{j}^{\ast }(s_{i})\subset E_{j}$ for some $j\neq i$ (i.e., $j$ cannot
support $s_{i}$ at $s^{\ast }$) and $E_{j}^{\ast }(s^{\ast })=E_{j}$ for
every agent $j$. Moreover, the agent loses less than $-1$ due to the
outcome, and at most $I-1$ to $\tau ^{3}$, and incurs no loss from $\tau
^{1} $ (the truth is not refutable) or $\tau ^{4}$ (all claims are supported
now). In conclusion, this is a profitable deviation.
\end{proof}

\begin{claim}
\label{c_SRL_4}No agent reports a self-refutable lie with positive
probability.
\end{claim}

\begin{proof}
Suppose to the contrary that agent $i$ reports $m_{i}=\left(
s_{i},E_{i}\right) $ where $s_{i}$ is a self-refutable lie.

Consider a deviation to the truth for agent $i$, i.e., consider a message $%
\tilde{m}_{i}$ which only replaces $s_{i}$ with $s^{\ast }$ in $m_{i}$. The
following table summarizes the payoff change by deviating from $m_{i}$ to $%
\tilde{m}_{i}$:%
\begin{equation*}
\begin{tabular}{|l|l|l|l|l|l|}
\hline
$g$ & $\tau _{i}^{1}$ & $\tau _{i}^{2}$ & $\tau _{i}^{3}$ & $\tau _{i}^{4}$
& In total \\ \hline
$>-1$ & $\geq 0$ & $\geq 0$ & $\geq I-1$ & $\geq 0$ & $>0$ \\ \hline
\end{tabular}%
\end{equation*}

In words, since $s_{i}$ is a self-refutable lie and $E_{i}^{\ast
}(s_{j})=E_{i}^{\ast }(s^{\ast })\neq E_{i}^{\ast }(s_{i})$, agent $i$ gains
at least $I-1$ from $\tau _{i}^{3}$ (wherein the first equality is from
Claim \ref{c_SRL_3}); moreover, the agent incurs a loss of at most $1$ from
the outcome and no loss or gains from other transfers. Hence, this is a
profitable deviation.
\end{proof}

\subsubsection{Implementation}

To sum up, it follows from Claims \ref{c_ORL_2}, \ref{c_NRL_3}, and \ref%
{c_SRL_4} that with probability one each agent reports the true state.
Hence, to prove implementation, we need only establish the following claim.

\begin{claim}
In equilibrium, no transfer is incurred.
\end{claim}

\begin{proof}
Since all state claims are truthful, it suffices to argue that all agents
present the tightest evidence in equilibrium. Indeed, if any agent is not
presenting the tightest evidence, then $\tau ^{4}$ is active with positive
probability. It then follows from Lemma \ref{l_tight} that deviating to the
tightest evidence is a profitable deviation.
\end{proof}

\subsection{Discussion}

We add a few remarks here. First, notice that from Lemma \ref{l_tight},
presenting additional evidence is never harmful, and is beneficial under
some cases. Therefore, it is a weakly dominant strategy to always present
all the evidence. It is clear that the above mechanism implements under
iterated elimination of weakly dominated strategies as well. This yields us
double implementation, in both mixed Nash equilibrium, and iterated
elimination of weakly dominated strategies.\footnote{%
If all agents present their tightest evidence, presenting an other-refutable
lie is dominated by presenting the truth owing to $\tau ^{1}$. Presenting
non-refutable lies is also dominated by presenting the truth due to $\tau
^{2}$. Any self-refutable lie which induces a evidence set tighter than that
under the truth for other agents is dominated by the truth due to $\tau ^{2}$%
, and then self-refutable lies are dominated by the truth due to $\tau ^{3}$.%
}

Second, in considering the above mechanism, it is evident that even though
the designer can impose transfers off the equilibrium, it is not sufficient
to simply penalize any profiles the designer finds undesirable. Rather, the
main challenge is to allow the agents a profitable deviation as well. The
idea can be seen from how the mechanism eliminates both nonrefutable and
self-refutable lies. For the elimination of nonrefutable lies, it is
critical that when an agent is being penalized for an unsupported
nonrefutable lie, the other agents must be presenting their tightest
evidence, so that the truth is actually supported, enabling the agent to
avoid the penalty by deviating to the truth. If this were not the case, then
in deviating to the truth, the agent's report will still be unsupported.
Coming to the elimination of self-refutable lies, it is critical that we
extract the article refuting a self-refutable lie from the agent himself, as
otherwise the cross-checks starts from the wrong profile of evidence, and we
would not be able to realign the profile towards the truth.

This leads to a somewhat counterintuitive design choice which we make during
the elimination of self-refutable lies. Notice that when an agent presents a
self-refutable lie, \emph{and} presents his tightest evidence, he actually 
\emph{refutes} his own state claim. This does not in fact lead to a penalty
for the agent, although penalizing an agent for being internally
inconsistent could be logical in certain circumstances.\footnote{%
For instance, in a partial implementation exercise of \cite{BDL2019} (p.
545), they begin with the understanding that agents must present the
maximal/tightest evidence associated with his state claim at the peril of
large punishments.} Our mechanism however is based on a series of
cross-checks (in Claims \ref{c_SRL_3} and \ref{c_SRL_4}) which only work
when we begin from the \emph{correct} profile of evidence. In particular,
the designer can incentivize the presentation of tightest evidence but is
not able to directly incentivize the presentation of the true state.%
\footnote{%
This differs from the classical implementation problem with preference
variation, where dictator lotteries ala \cite{AM94} can be used to elicit an
agent's preference.} In its simplest form, our mechanism takes the following
position - \textquotedblleft when there is something wrong with the message
profile, prioritize obtaining the tightest profile of evidence over all
else\textquotedblright . This allows the cross-checks which realign the
profile towards the truth.

\subsection{The Role of Normality}

\label{Norm}

We now turn to the role of normality in our setup. When agents cannot submit
arbitrary amounts of evidence, the necessary condition is obtained by \cite%
{KT2012} in their Propositions 2 and 3. First, for any state $s$, define by $%
T^{f}(s)=\{s^{\prime }:f\left( s^{\prime }\right) \neq f\left( s\right) \}$
the set of states in which the desirable outcomes differ from $f\left(
s\right) $. Recall that we wish to implement without relying on preference
variation. In particular, under constant preferences, Proposition 3 in \cite%
{KT2012} requires that $s$ and $T^{f}(s)$ be distinguishable, which means
that either some agent can refute $s$ at any state in $T^{f}(s)$ or refute
every state in $T^{f}(s)$ at $s$ using a single article of evidence. In the
appendix, we state and prove Theorem 7, which shows that this condition is
also sufficient for mixed-strategy implementation with two or more agents in
a finite albeit indirect mechanism which requires submission of two articles
of evidence.

\subsection{Budget Balance}

\label{BB}

We will now establish that the above mechanism can be modified to achieve
budget balance. The major challenge in achieving this goal is the
redistribution of penalties to other agents without affecting the incentives
of the recipient agents. In general, this cannot be achieved with two
agents, as we will show in Section \ref{RP}. In the following discussion, we
consider a setting with at least three agents.

One way to transform a two agent unbalanced mechanism into a three agent
balanced mechanism is to choose two agents at random to play the unbalanced
mechanism (with the transfers for an agent appropriately scaled up to
reflect the probability of being chosen) and redistribute the transfers to
the third agent. While this approach is quite general, it requires a
stochastic mechanism. In what follows, we provide modifications to the above
mechanism which achieves budget balance without randomization.

First, it is clear that $\tau ^{1}$, the transfer for the elimination of
other-refutable lies is already budget balanced. The incentive to present
all the evidence stemming from $\tau ^{4}$ is key to the removal of both
self-refutable and nonrefutable lies. This incentive can be arbitrarily
small, and is only active when some agent's claim is not supported by the
other agents. Redistributing this small incentive to the agent whose state
claim is unsupported does not affect his incentives since the penalty from
his state claim not being supported is much larger. The second transfer $%
\tau ^{2}$, which penalizes an agent for his state claim being unsupported
is redistributed among the other agents, with a minor modification - $\tau
_{i}^{2}$ is redistributed evenly to only those other agents $j\neq i$ who
have supported $i$'s state claim $s_{i}$, and redistributed evenly to all
agents if no other agents have supported $s_{i}$. The third transfer $\tau
^{3}$ is directly redistributed evenly to all agents. The resulting
mechanism is, therefore, budget balanced. We refer the reader to Appendix %
\ref{app:BB} for a discussion of how implementation is obtained under this
modified mechanism.

\subsection{\label{SmallTransfers}Implementation with Small Transfers}

While the transfers involved in the mechanism above have been imposed only
off the equilibrium, the transfers are \textquotedblleft
large\textquotedblright\ since they need to dominate the agents' utility
differences from outcomes. In reality, agents might not be willing or able
to pay these fines. Here we present a result for implementation with
arbitrarily small off-the-equilibrium transfers, as long as the designer can
randomize and there are at least three agents. To this end, we construct an
indirect mechanism building on similar ideas from \cite{AM94}. We begin by
defining the appropriate notion of implementation prevalent in the
literature for this case.

\begin{definition}
An SCF is Nash implementable with arbitrarily small transfers if for any $%
\varepsilon >0$, there is a mechanism $\mathcal{M}=(M,g,\left( \tau
_{i}\right) _{i\in \mathcal{I}})$ such that for any profile of utility
functions $v=\left( v_{i}\right) _{i\in \mathcal{I}}$, any state $s$, and
any mixed-strategy equilibrium $\sigma $ of the game $G\left( \mathcal{M}%
,v,s\right) $, we have $g(m)=f(s)$ and $\tau _{i}(m)=0$ for each message
profile $m\in $supp$\sigma \left( s\right) $ and the total
(off-the-equilibrium) transfer to any agent can be limited to being no
greater than $\varepsilon $.
\end{definition}

We now state the formal result as follows:

\begin{theorem}
\label{small}Suppose the evidence structure is given by $\mathcal{E}%
_{i}\left( \cdot \right) $. If there are at least three agents, then an SCF $%
f$ is Nash-implementable with arbitrarily small transfers if and only if it
is measurable with respect to $\mathcal{E}_{i}\left( \cdot \right) $.
\end{theorem}

Intuitively, this result is based on the following ideas. A lottery is used
to \emph{divide} the incentive for manipulating the outcome into $K$ parts
(called rounds), where $K$ can be chosen as large as necessary to meet the
transfer bound required. The first round uses the implementing mechanism
with its transfers scaled down, since only a small part of the outcome is
controlled by it. This allows for the true state to be revealed.\footnote{%
Note that this is not dependent on the mechanism used. Any fully
implementing mechanism can be adapted into this framework, a fact we will
leverage later when studying other setups.} In each of the following rounds,
agents are incentivized to agree with the unanimous first round report of
the truth, failing which the first deviant is penalized an amount that is
small enough that it meets the transfer bound, yet large enough that it
dominates the incentive for manipulating the outcome of the round. This
penalty only applies to the first round with disagreement, as repeating it
will lead to a large transfer. It is sufficient for implementation however,
because no agent wants to be the first to deviate in any round. For further
details, we refer the reader to the formal proof of Theorem \ref{small} in
Appendix \ref{app:small}.

\section{\label{COSTLY}Costly Evidence}

So far, we have studied hard evidence, which corresponds to the notion that
evidence which is available to the agent is costless to present. In this
section, we relax the assumption that presenting evidence is costless.
Evidentiary costs are mentioned as an important area of further research in
both \cite[footnote 10]{BW2007} and \cite[p. 1714]{BL2012}.

With the above motivation, we derive an extension of our implementation
result to a setting with costly evidence. More formally, the environment is
the same as that in Section \ref{Env}, except for the addition of a cost
function $c_{i}:\mathcal{E}_{i}\times S\rightarrow \mathbb{R}_{+}$ which is
bounded by a (possibly) large positive cost $C$. Here, we note that we allow
the evidentiary cost to depend on the state.

There are two possible stances on the designer's knowledge of the cost
structure. Either the designer does not know $c_{i}\left( \cdot \right) $,
and only knows $C$ (so that he is unable to exploit the variation of costs
between states), or he knows $c_{i}\left( \cdot \right) $ as well (whereupon
he can exploit the variation of cost among states). We treat these two cases
separately in the following sections.

\subsection{Implementation regardless of Cost Variation}

\label{costrobust}

In this section, we assume that $c_{i}\left( \cdot \right) $ is common
knowledge among the agents but the designer only knows $C$. Then, the
designer cannot exploit cost variation. The definition of normality remains
the same as in Section \ref{Env}. Further, the notion of implementation
remains the same as that in Section \ref{Model}, so that the designer is
indifferent to the cost of evidence submission. We obtain the following
result:

\begin{theorem}
\label{hardcostly}An SCF $f$ is Nash implementable in a direct revelation
mechanism for every costly evidence structure $\mathcal{E}_{i}\left( \cdot
\right) $ if and only if it is measurable.
\end{theorem}

We provide the proof of Theorem \ref{hardcostly} along with a detailed
sketch in Appendix \ref{app:hardcostly}. In the proof, we adopt the same
direct revelation mechanism constructed in Section \ref{Mech} but suitably
adjust the relative scales of the four transfer rules. The proof requires a
substantially different approach from that of Theorem \ref{main}. Due to the
evidentiary cost, even after raising the transfers, lies can be eliminated
only with high probability rather than probability one. This turns out to be
sufficient to ensure that the agents present their tightest evidence in
order to support the truthful state claim which is presented with high
probability. However, once the tightest evidence is presented, the agents
can no longer lie in their state claims even with small probability, so that
implementation is achieved.

\subsection{Implementation under Cost Variation}

In this section, we assume that $c_{i}\left( \cdot \right) $ is common
knowledge among the designer and the agents. A mechanism, therefore, can
depend on the cost structure to eliminate incorrect claims of state. This
setup corresponds to that in \cite{KT2012}. We emphasize that in this setup
articles of evidence do not necessarily associate with subsets of the state
space because of the element of cost variation. Following \cite{KT2012}, we
define the set of cheapest evidence in any state as $\mathcal{E}%
_{i}^{l}(s)=\arg \min_{E_{i}}c_{i}(E_{i},s)$. We also normalize the costs so
that the difference in costs between any two articles of evidence between
any two states is less than $1$ dollar. That is, $c(\cdot )$ is normalized
such that $\left\vert c_{i}(E_{i},s)-c_{i}(E_{i},s^{\prime })\right\vert <1$
for any $i$, $E_{i}$, $s$ and $s^{\prime }$. We present the notion of
implementation we work with below.

\begin{definition}
\label{implementation-costly}An SCF $f$ is directly Nash-implementable in
pure (resp. mixed) strategies if there is a mechanism $\mathcal{M}=(S\times 
\mathcal{E},g,\left( \tau _{i}\right) _{i\in \mathcal{I}})$ such that for
any profile of bounded utility functions $v=\left( v_{i}\right) _{i\in 
\mathcal{I}}$, any state $s$, any pure (resp. mixed) strategy Nash
equilibrium $\sigma $ of the game $G\left( \mathcal{M},v,s\right) $, and any
message profile $(s,E)$ in the support of $\sigma (s)$,
\end{definition}

\begin{enumerate}
\item[(i)] $g(m)=f(s)$ and $\forall i$, $\tau _{i}(m)=0$;

\item[(ii)] $E\in \mathcal{E}^{l}(s)$
\end{enumerate}

As in the hard evidence setting, implementation requires that for any
profile of bounded utilities, and in each equilibrium of the implementing
mechanism, the outcome be $f$-optimal and the transfer to each agent be
zero. In addition, it also requires that only an article from the set of
cheapest evidence be submitted in equilibrium. In this section, we focus on
pure strategy equilibria, and discuss a treatment of mixed equilibria in
Section \ref{LimEv}.

\subsubsection{Evidence Monotonicity}

\cite{KT2012} establish that a condition called \emph{evidence monotonicity}
is necessary for implementation in the above setup using a mechanism which
only admits the submission of cheapest evidence in equilibrium.\footnote{%
Notice that Theorem \ref{hardcostly} does not involve this restriction. We
discuss the implications in Section \ref{LimEv}.} Since we wish to implement
while maintaining robustness to agents' utility functions, constant
preferences is a possible scenario under which a mechanism must still
implement. If we allow for transfers, then, we obtain the following
characterization of evidence monotonicity under our setting.

\begin{definition}
An SCF $f$ is evidence-monotonic under constant preferences if there exists $%
E^{\ast }:S\rightarrow \mathcal{E}$ such that
\end{definition}

\begin{enumerate}
\item[(i)] for all $s$, $E^{\ast }(s)\in \mathcal{E}^{l}(s,f(s))$

\item[(ii)] for all $s$ and $s^{\prime }$, if

$\forall i,t\in \mathbb{R},E_{i}^{\prime }:[-c_{i}(E_{i}^{\ast }(s),s)\geq
t-c_{i}(E_{i}^{\prime },s)\implies -c_{i}(E_{i}^{\ast }(s),s^{\prime })\geq
t-c_{i}(E_{i}^{\prime },s^{\prime })]$,

then $f(s)=f(s^{\prime })$.
\end{enumerate}

Alternatively, if $f(s)\not=f(s^{\prime })$, then $\exists
i,t,E_{i}^{^{\prime }}$ such that $c_{i}(E_{i}^{\ast }(s),s)\leq
c_{i}(E_{i}^{\prime },s)-t$ but $c_{i}(E_{i}^{\ast }(s),s^{\prime
})>c_{i}(E_{i}^{\prime },s^{\prime })-t$. This yields $c_{i}(E_{i}^{\prime
},s^{\prime })-c_{i}(E_{i}^{\ast }(s),s^{\prime })<c_{i}(E_{i}^{\prime
},s)-c_{i}(E_{i}^{\ast }(s),s)$. The implementing mechanism which we present
later will be direct. In context of a direct mechanism, we interpret $i$'s
action of submitting $(E_{i}^{\prime },s^{\prime })$ instead of $%
(E_{i}^{\ast }(s),s)$ as a challenge to the state claim of $s$ at $s^{\prime
}$, and denote $t$ as the (possibly negative) challenge reward and $%
E_{i}^{\prime }$ as the challenge evidence. In essence, agent $i$ credibly
informs the designer that the state is not $s$ by asking for an amount of
money $t$ for presenting an article of evidence $E_{i}^{\prime }$ instead of 
$E_{i}^{\ast }(s)$. This is profitable if the state is $s^{\prime }$ and not
profitable if the state is indeed $s$. This implies that if $i$ is to
challenge $s$ at $s^{\prime }$, then there is an article $E_{i}^{\prime }$
which has become cheaper relative to $E_{i}^{\ast }(s)$ in going from $s$ to 
$s^{\prime }$. We define $\mathcal{E}_{i}^{\gamma }(s,s^{\prime })$ as the
set of \emph{challenge evidence} for agent $i$ when challenging $s$ at state 
$s^{\prime }$, and make an arbitrary selection $E_{i}^{\gamma }(s,s^{\prime
})$ from it, which is used in the implementing mechanism.

\subsubsection{A Classification of Lies}

We begin with a classification of lies (for some true state $s^{\ast }$) and
then use a direct mechanism to eliminate them in sequence. To do so, we
first formalize the notion of refutability in this setup.

\begin{definition}
An agent $i$ can challenge a state claim $s$ when $\exists (s_{i},E_{i})$
such that $c_{i}(E_{i},s_{i})-c_{i}(E_{i}^{\ast
}(s),s_{i})<c_{i}(E_{i},s)-c_{i}(E_{i}^{\ast }(s),s)$.
\end{definition}

Note that there is a difference between the interpretation of challenge
between the usual hard evidence setting and this setup. In the hard evidence
setting, when an agent presents an article of evidence $E$ which does not
contain a state $s$, he definitively refutes the state $s$. In this setting
however, articles of evidence cannot be associated with a subset of the
state space in the same way, since the "reversal" which credibly signals to
the designer that the state is not $s$ requires the commitment of a certain
sum of money in exchange for the challenge evidence. The inequality in the
definition above assures us of the existence of a sum of money that enables
this reversal. In essence, whereas refutability is a property of the setup
under hard evidence, a mechanism is required for the same in this setup.

\subsubsection{Implementation}

We now present our main result for this setting.

\begin{theorem}
\label{main_soft}Suppose that $\mathcal{E}_{i}\left( \cdot \right) $ is a
costly evidence structure. Then, an SCF $f$ is directly Nash-implementable
in pure strategies in a direct revelation mechanism if and only if it is
evidence-monotonic under constant preferences.
\end{theorem}

We draw attention to a few interesting points regarding the above result
here. First, we note that while evidence monotonicity is necessary
irrespective of the mechanism used, the mechanism we present is direct.
Second, \cite{CKSX2021} contains an example which establishes that it is not
possible to obtain direct implementation of some Maskin-monotonic social
choice functions where there are only two agents.\footnote{%
The example discusses rationalizable implementation, but it also applies to
Nash implementation.} The impossibility essentially derives from the
difficulty of figuring out which agent is challenging the other when two
agents disagree in their state claims. The presence of evidence allows us to
bypass this difficulty.\footnote{%
See Lemma \ref{no_chal_cheapest} in Appendix \ref{app:softcostly} for more
details.} Third, while the result states that evidence monotonicity under
constant preferences is necessary, we allow for variation of preferences in
the proof of Theorem \ref{main_soft}. The necessity of this condition arises
out of our desire to achieve implementation regardless of preference
variation. If the designer cannot exploit preference variation,
implementation must obtain from variations in the cost of evidence.

\subsubsection{Implementation in mixed strategies}

For implementation in mixed-strategy Nash equilibrium which accounts for
general cost variation (Definition \ref{implementation-costly}), we also
present an alternative treatment in Appendix \ref{CEMS}. In this treatment,
we impose a stronger condition than evidence monotonicity under constant
preferences. Specifically, we require that for any two states $s$ and $%
s^{\prime }$ with distinct social outcomes, there be at least one agent for
whom an article of evidence that was not cheapest under $s$ is now cheapest
under $s^{\prime }$. We term this condition evidence monotonicity* and prove
mixed-strategy implementation (also in a direct mechanism) under evidence
monotonicity$^{\text{*}}$.\footnote{%
If we only allow implementing mechanisms which use arbitrarily small rewards
to the agents, then evidence monotonicity* becomes necessary for
implementation (this is established in Theorem 6). Our implementing
mechanism in this case indeed satisfies the requirement.} The result
complements Theorem \ref{hardcostly} in requiring that only one article of
cheapest evidence be submitted in equilibrium, but not normality.\footnote{%
We also conjecture that Theorem \ref{hardcostly} can be established without
requiring normality, by making use of an indirect mechanism akin to the
implementing mechanism in the proof of Theorem 7.}

\subsection{\label{LimEv}Evidence Monotonicity versus Measurability}

Theorems \ref{hardcostly} and \ref{main_soft} approach costly evidence in
different ways. It is natural to ask how these results compare. In the
costly evidence setting, some articles which are unavailable are considered
to have an infinite cost. In such a setting, measurability equates to the
requirement that the set of evidence with finite cost must change between
states with different social outcomes. This requirement is strictly stronger
than evidence monotonicity. To see this, consider Example 2 of \cite{KT2012}
(henceforth KT) which derives from agents having a preference for honesty.
The example can be viewed as a case of costly evidence in which all states
have the same evidence sets, but the cheapest evidence is distinct in each
state. KT show that this makes any social choice function evidence
monotonic. However, since there is no variation in the set of evidence with
finite cost, only constant social choice functions are measurable.

KT show in their Corollary 4 that evidence monotonicity coincides with
measurability in a hard evidence structure (where evidence is costless when
it is available) which satisfies normality.\footnote{%
Note that this only holds among pairs of states which do not satisfy maskin
monotonicity, for instance with state independent preferences since
otherwise, evidence monotonicity may be satisfied via preference variation.}
In contrast, our Theorem \ref{hardcostly} allows for hard evidence which is
available and yet has a positive cost. In this setting, even when the
evidence structure satisfies normality, there can still be social choice
functions which satisfy measurability but not evidence monotonicity (under
constant preferences). See Appendix \ref{app:meas_but_not_em} for a detailed
description. Such SCFs can therefore be implemented according to Theorem \ref%
{hardcostly} but not Theorem \ref{main_soft}.

The discrepancy arises because the implementation notion in Theorem \ref%
{main_soft} requires that only one article of cheapest evidence be submitted
in equilibrium. Apparently, the requirement must be associated with a fixed
costly evidence structure and thereby muted when we demand implementation
regardless of evidentiary cost variation, as we do in Theorem \ref%
{hardcostly}. In this regard, Theorem \ref{hardcostly} is a step towards
answering a question which KT pose as to which social choice functions could
be implemented if the designer allowed for the presentation of costly
articles of evidence in order to elicit information from the agents.%
\footnote{\cite[p.349]{KT2012} cite screening as an example for why it might
be interesting to allow for costly evidence provision in equilibrium.}

\section{\label{RP}Incomplete Contracts and Renegotiation-Proofness}

In this section, we apply our results under hard evidence to revisit a
classical issue in contract theory -- that of bilateral contracting with
observable but unverifiable information. This problem arises when two
parties wish to condition a contract on certain state variables which are
commonly observable, but not verifiable by a third party, such as a court.
Attempting to directly condition a contract on these variables runs into
difficulties, because whoever is responsible for enforcing the contract may
not be able to ascertain which state occurred and thus may also be unable to
resolve any disputes between the agents.

Implementation theory deals with issues of contract design in such
situations by exploiting the idea that contracts can be made contingent on
the messages reported by the parties to make the observable state
verifiable. In particular, by designing suitable revelation mechanisms for
the contracting parties, it is often possible to achieve the same outcomes
as those arrived at with fully contingent contracts. However, most papers in
this literature use mechanisms of the Moore-Repullo type (see, for example, 
\cite{MM99}, \cite{MT99}, and \cite{M2002}), which are vulnerable to
renegotiation among the agents, because they involve penalizing both agents
in certain off-equilibrium paths. In general, the possibility of
renegotiation significantly limits the set of implementable outcomes. For
instance, \cite{MM99} (henceforth MM) considers an example in which a seller
(she) owns a product which a buyer (he) is interested in obtaining. The
seller has an option to make an investment in the product at cost $c$ to
raise its value from $\phi $ to $\theta $ for the buyer. The investment is
efficient (i.e., $c<\theta -\phi $), observable by the two parties, but
unverifiable. In this setting, when the cost of the investment is more than
half of the value it adds to the good, MM establish that the only
renegotiation-proof contract is a null contract. This is because the buyer
refuses to accept the good and renegotiates the price outside of the
contract. That is, it is not possible to incentivize efficient investment in
this scenario.

In the context of the example, it becomes clear why nonverifiability poses
an issue. If the investment were verifiable, a mechanism could fine the
buyer for refusing to trade after the investment had been made, and transfer
the fine to the seller, thus preventing the buyer from lying, and indeed
also preventing renegotiation. Even without verifiability, if the seller
could \emph{prove} that she has indeed made the investment or refute the
claim that she has not made the investment, a court could fine the buyer (if
he refuses to trade) and reward the seller based on this proof. This leading
example motivates our choice of using the hard evidence results in this
context, as it is more natural that evidence in contracting situations will
take the form of subsets of the state space, refuting alternative states.
Arguably, it is not possible to prove an investment which one has not made,
at whatever cost. This immediately leads to the question -- what conditions
must the hard evidence structure satisfy so that renegotiation-proof
contracts can be made?

In order to address this question, we begin by formalizing the notion of
renegotiation. Following MM, we assume $A$ is a finite set, and define $T$
as the set of transfers with budget surplus.\footnote{%
That is, $T=\{t\in \mathbb{R}^{I}$ s.t. $\sum_{i}t_{i}\leq 0\}$. Since this
is a contract between agents, it is not possible to find money outside the
contract to finance a budget deficit.} With this, we define the
renegotiation process via a renegotiation function $h:A\times T\times
S\rightarrow A\times T$, where $T$ defines the space of transfers to\emph{\ }%
the $I$ agents. This function can be thought of as a transformation on the
outcome and profile of transfers of a mechanism $\mathcal{M}=(M,g,\left(
\tau _{i}\right) _{i\in \mathcal{I}})$, which occurs before the agents
evaluate the outcome using their utility functions $u_{i}$. That is, given a
mechanism $\mathcal{M}$, agents submit their messages to the mechanism,
which yields an outcome and a set of (budget surplus) transfers; agents then
renegotiate this combination of outcomes and transfers to another
combination, and then evaluate their utilities. We think of this in terms of
the game $G\left( \mathcal{M},v,h,s\right) $, which differs from the game $%
G\left( \mathcal{M},v,s\right) $ in that instead of the outcome being
defined as $g(m)$ and the transfer profile as $\tau (m)$, the outcome is
defined as $h^{a}(g(m),\tau (m),s)$ and the transfer profile as $%
h^{t}(g(m),\tau (m),s)$. Now, we define the notion of efficiency for an
allocation.

\begin{definition}
\label{eff}An allocation $(a,(t_{i})_{i\in \mathcal{I}})$ is efficient with
respect to a profile of utilities $v=(v_{i})_{i\in \mathcal{I}}$ at state $s$
if there does not exist $(\hat{a},\hat{t})\in A\times T$ such that for every
agent $i$, $u_{i}(\hat{a},s,\hat{t}_{i})\geq u_{i}(a,s,t_{i})$ with strict
inequality for some $i$.
\end{definition}

Following MM, we make the following three assumptions about the
renegotiation function $h$. First, we assume that the renegotiation function 
$h$ is \emph{predictable}, which essentially amounts to saying that $h$ is
common knowledge among agents and deterministic. Second, we assume that $h$
is \emph{individually rational}, that is, if at every state, all agents
weakly prefer the renegotiated outcome to the original one. This is a
natural restriction as no agent can be forced into renegotiation. Finally,
we assume that renegotiation is \emph{efficient}, which means that $h\left(
\cdot ,\cdot ,s\right) $ results in efficient allocations at every $s$.

In what follows, we will constrain the scope of the discussion to
combinations of $f$, $v$, and $h$ such that $f(s)$ is efficient with respect
to $v$, and $h$ satisfies the properties described above.

\begin{definition}
A social choice function $f$ is implementable with renegotiation in Nash
equilibrium if there is a mechanism $\mathcal{M}=(M,g,\left( \tau
_{i}\right) _{i\in \mathcal{I}})$ such that for any state $s$, any profile
of utilities $v$, and any mixed-strategy Nash equilibrium $\sigma =(\sigma
_{i})_{i\in \mathcal{I}}$ of the game $G\left( \mathcal{M},v,h,s\right) $,
we have $h(g(m),\tau (m),s)(\hat{a})=f(s)$ and $h(g(m),\tau (m),s)(\hat{t}%
)=0 $ for every message profile $m$ on the support of $\sigma $.
\end{definition}

In the spirit of this paper, we require that implementation obtain
regardless of the utility functions $v$ (subject to the constraints
mentioned above). We now turn to characterizing the necessary and sufficient
conditions for renegotiation-proof implementation.

\begin{theorem}
\label{RPT_REINTERPRET}Assume that $I=2$ and $\mathcal{E}$ is the evidence
structure. An SCF $f$ is implementable with renegotiation in Nash
Equilibrium if and only if for any pair of states $s$ and $s^{\prime }$ in $%
S $ such that $f(s)\neq f(s^{\prime })$, one of the following is true:

\begin{enumerate}
\item[(a)] There is one agent who can refute $s^{\prime }$ at $s$ and $s$ at 
$s^{\prime }$; or

\item[(b)] Both agents can refute $s^{\prime }$ at $s$ and neither of them
can refute $s$ at $s^{\prime }$.
\end{enumerate}
\end{theorem}

For the formal proof, we refer the reader to Appendix \ref{app:RP}.

To illustrate the above conditions, we consider an alternate evidence
structure in the example from MM\ above -- what if the buyer is the only
agent who can prove that the investment was made (or not)? In practice, this
may often be the case, for instance if the buyer has some sort of (private)
suitability test which can check if the investment has been made. Theorem %
\ref{RPT_REINTERPRET} tells us that a renegotiation-proof contract exists in
this case. This is because the burden of proving that the investment was not
made also falls to the buyer, and this is impossible when the investment has
been made. This, in turn, means that any set of prices can be implemented
with this evidence structure (note that any price is Pareto efficient),
including, of course, the set of prices which recoup the cost of the
investment for the seller, thereby incentivizing her to make the efficient
investment. In this context, we refer the reader to Appendix \ref%
{app:RPEvidenceEx} for an example of an evidence structure which satisfies
measurability but not the conditions of Theorem \ref{RPT_REINTERPRET}. While
the Conditions (a) and (b) are fairly demanding in this example, requiring a
fair bit of provability by the buyer, they are nevertheless \emph{necessary}%
, suggesting that renegotiation-proof implementation is significantly more
challenging than the notion in Definition \ref{implementation}.

There are a few points of interest we wish to clarify here. First, we note
that the characterization in Theorem \ref{RPT_REINTERPRET} is stronger than
measurability. Indeed (e1) of Definition \ref{e1e2} and either one of
Condition (a) or (b) yields that the agent who can refute $s$ at $s^{\prime
} $ has an article of evidence at $s$ which he does not have at $s^{\prime }$%
. Then, $f$ trivially satisfies measurability. Second, the mechanism
proposed in the proof of Theorem \ref{RPT_REINTERPRET} needs only a direct
revelation mechanism to work, and at least one of the Conditions (a) or (b)
are necessary irrespective of the mechanism used.

We discuss the intuitions behind the conditions presented in Theorem \ref%
{RPT_REINTERPRET} now. We prove the necessity of the conditions by
constructing a pair $v,h$ under which\ a) at least one of the agents has
incentive to lie, and b) the agent in question \emph{can} lie, regardless of
the implementing mechanism.

In proving sufficiency, we construct a direct revelation mechanism with two
transfers: a penalty of 1 dollar for agent $i$ if he is unable to support
his state claim and a penalty of 2 dollars if his state claim is refuted by
the evidence which agent $j\neq i$ presents. Each fine is paid to the other
agent, so that the mechanism has a balanced budget.\footnote{%
Recall that the maximum bound of utility is less than 1 dollar, so that
these transfers dominate the utility from the outcome.}

Assume that the true state is $s^{\ast }$. First, no agent reports
other-refutable lies with positive probability, because the other agent
presents the refuting evidence with probability 1; this loses the first
agent 2 dollars, which is more than the utility difference from the outcome
and the other transfer. Second, if an agent reports a self-refutable (resp.
nonrefutable) lie $s$ with positive probability, then Condition (a) (resp.
Condition (b)) of Theorem \ref{RPT_REINTERPRET} implies that this agent can
also refute $s^{\ast }$ at $s$. Therefore, he cannot support his state claim
and will be fined 1 dollar. Since this is more than the utility difference
from the outcomes, deviating to $(s^{\ast },E_{i}^{\ast }(s^{\ast }))$ is
profitable. Therefore, the only equilibrium is truth-telling.

To summarize, we have established the necessary and sufficient condition for
renegotiation-proof Nash implementation in settings with hard evidence in
the language of refutability. This allows us to answer an important question
-- it is well known that ex-post verifiability is sufficient for a
renegotiation-proof contract, but \textquotedblleft how much
provability\textquotedblright\ is really required to ensure that such
contracts exist? Theorem \ref{RPT_REINTERPRET} offers an answer.

\section{Related Literature}

\label{literature}

This paper contributes to the literature on implementation with evidence.
For early work in this area, we refer the reader to \cite{GL1986}, \cite%
{BW2007} and \cite{DS2008}. See also \cite{DS2007} for a study of partial
implementation with costly evidence and \cite{KT2012} for a survey of other
early works on implementation with evidence.

Our implementation exercise is most closely-related to \cite{BL2012} and 
\cite{KT2012}. \cite{BL2012} present two main results. First, they achieve
subgame-perfect Nash implementation using a perfect-information mechanism
with large off-the-equilibrium transfers (which achieve budget balance when
there are three or more agents). This result does not require normality or
integer games. Then, \cite{BL2012} achieve Nash implementation with three or
more agents by using small off-the-equilibrium transfers. This latter result
does require normality and integer games. In contrast, our Theorem \ref{main}
requires normality, works with two agents, and employs a direct revelation
mechanism to achieve (mixed-strategy) Nash implementation without integer
games but with large off-the-equilibrium transfers (which also achieve
budget balance when there are three or more agents). Neither Theorems 1 or 2
of \cite{BL2012} account for evidentiary cost which our Theorem \ref%
{hardcostly} deals with.\footnote{%
In their Theorems 4 and 5, \cite{BL2012} show that their results are robust
to a variety of preference and belief specifications. Our results are
consistent with these, in the sense that the designer does not need to know
anything about the preferences of the agents except for the uniform bound on
utilities, although we maintain the complete-information assumption
throughout.}

\cite{KT2012} studies the costly evidence setting. They use a canonical,
Maskin-style mechanism which relies on integer games and does not use
transfers. With three or more agents, their achieve pure-strategy Nash
implementation for every evidence-monotonic SCF.\footnote{%
Their Theorem 2 can be modified to account for mixed Nash equilibria as
well, using methods from \cite{KT2012Mixed} but this depends on integer
games as well.} Their evidence monotonicity notion allows for preference
variation and reduces to Maskin-monotonicity when all messages are cheap
talk. In contrast, Theorem \ref{main_soft} makes use of a direct revelation
mechanism without integer games to achieve pure-strategy implementation with
two or more agents. Further, to focus on evidentiary cost variation, we
adopt the notion of evidence monotonicity under constant preferences under
which our implementation result disregards (and therefore is also robust to)
preference variation.

Another related strand of literature is that of implementation with
preferences for honesty, wherein it is assumed that agents prefer to tell
the truth if they do not gain from lying. Preferences for honesty can be
viewed as a case of costly evidence in which all states have the same
evidence sets, but the least-cost evidence is distinct in each state. \cite%
{KTH2014} establish that with two or more agents, preferences for honesty
and a condition called separable punishments (which generalize
off-the-equilibrium transfers), any SCF\ can be implemented in a finite but
indirect mechanism without integer or modulo games. \cite{DS2011} also
establish that with preferences for honesty all social choice
correspondences which satisfy no veto power can be implemented by a
mechanism which uses integer games. Our Theorems \ref{hardcostly} and \ref%
{main_soft} apply to a general costly evidence structure beyond the specific
setting of preference for honesty and both invoke only direct revelation
mechanisms.\footnote{%
As \cite{KT2012} point out, as long as in any state at least one agent has a
preference for honesty, any SCF\ is evidence monotonic.}

\cite{MAM} also provides an implementation result using transfers that
allows for implementation of Maskin-monotonic social choice functions
without using integer games. A main focus of our exercise though is to
handle implementation with state-independent or consistent preferences by
making use of evidence. Our emphasis on evidence also allows for the novel
classification of lies which we propose and in turn allows for the new
approach to implementation which we present here. This is a feature
unavailable for implementation exercises which solely rely on preference
variation. Moreover, our treatment of the two-agent case extends to dealing
with renegotiation (wherein we provide a necessary and sufficient condition
for renegotiation-proof implementation), a pertinent issue when two-agent
implementation is used in contracting.

\section{Conclusion}

In this paper, we present full implementation results in settings with hard
and costly evidence. In the hard evidence setting, using a novel
classification of lies according to their refutability, we construct a
direct revelation mechanism which implements any measurable SCF with respect
to the evidence structure. Our mechanism invokes neither integer nor modulo
games, requires only two agents, accommodates evidentiary costs, and can
also be modified to account for limited solvency of the agents. Based on our
classification of lies, we also derive a necessary and sufficient condition
for the existence of bilateral renegotiation-proof contracts (\cite{MM99}).
In the costly evidence setting, we provide a mechanism that yields pure
strategy implementation in a two agent setting without using integer or
modulo games. The classification of lies approach is used here as well.

Our exercise leaves a number of open questions for future research. In
particular, \cite{BCS2020} pursue an extension of our exercise to an
incomplete-information setting with a commonly known state and uncertainty
among agents along the evidence dimension. Other directions include direct
implementation without transfers or structural assumptions which serve the
same role.

\appendix

\section{Appendix}

\label{appendix}

In this Appendix, we provide the details and proofs which are omitted from
the main body of the paper.

\subsection{\label{app:BB} Budget Balance}

Since $\tau ^{1}$ is budget balanced at the outset, Claims \ref{c_ORL_1} and %
\ref{c_ORL_2} are not affected by the modifications. Lemma \ref{l_tight}
continues to hold, because the only possible scenario under which providing
additional evidence may cause a loss to an agent $i$ (owing to the
redistribution of transfers) is when providing additional evidence supports
an agent $j$'s state claim $s_{j}$, and therefore causes losses from the
redistribution of $\tau _{j}^{2}$. However, for this to be the case, agent $%
k $ must then be supporting $s_{j}$. In this case, agent $i$ does not
actually receive any part of the redistributed revenues from $\tau _{j}^{2}$.

Claims \ref{c_NRL_2} and \ref{c_NRL_3} remain unaffected. If an agent $i$
presents a nonrefutable lie, we establish that all other agents still
present the tightest evidence. First, since $\tau ^{4}$ is only
redistributed to $i$, it does not affect other agent's incentives. Further,
an agent is only unwilling to support another agent (to avoid redistributive
losses from $\tau ^{2}$) if he was the only agent who did not support $s_{i}$
(so that in supporting $i$, he would switch off $\tau _{i}^{2}$). But, due
to the way that $\tau _{i}^{2}$ is redistributed, and the existence of a
third agent, he does not share in the redistribution of $\tau _{i}^{2}$ in
this case. Then, agent $i$ still faces a large fine (the redistributed
revenue from $\tau ^{4}$ is small) and chooses to deviate to $(s^{\ast
},E_{i}^{\ast }(s^{\ast }))$.

Claims \ref{c_SRL_2} through \ref{c_SRL_4} continue to hold as well. Since
agents are limited to presenting either the truth or self-refutable lies,
from the above argument, all evidence must still be tightest. If an agent $i$
plays a self-refutable lie which implies a smaller evidence set for $j$ than
the truth, then it is not supported, and agent $i$ prefers to deviate to the
truth as above. Therefore, all claims are consistent with other agents'
tightest evidence. The third transfer, $\tau ^{3}$ is a penalty for an agent 
$i$ for implying a different evidence set for himself than that implied by
others for him. His incentive for avoiding this penalty is not affected by
its redistribution to other agents. Therefore, with three or more agents,
this mechanism implements with budget balance.

\subsection{\label{app:small} Proof of Theorem \protect\ref{small}}

\subsubsection{Message Space}

The message space is augmented with $K+1$ additional claims of state, where
we call $K$ the number of rounds and it is chosen according to the upper
bound of allowable transfers. More precisely, the message space is as
follows: 
\begin{equation*}
m_{i}=(s_{i}^{0},E_{i},s_{i}^{1},s_{i}^{2},...,s_{i}^{K+1})\in M_{i}=S\times 
\mathcal{E}_{i}\times S\times ...\times S\text{ (K+1 times).}
\end{equation*}

\subsubsection{Outcome}

The outcome is specified as follows. \noindent Define $\rho ^{k}(m)$ (for $%
k=2,...K+1$) as follows:

\begin{equation*}
\rho ^{k}(m)=\left\{ 
\begin{tabular}{ll}
$f(s)$ & $\text{ if\thinspace\ }\exists s\text{ s.t. }|\{j:s_{j}^{k}=s\}|%
\geq I-1$ \\ 
$b$ & otherwise%
\end{tabular}%
\right.
\end{equation*}%
\noindent where $b$ is an arbitrary lottery over $A$. Then, the outcome of
message profile $m$, denoted by $\bar{g}(m)$ is defined as follows.

\begin{equation*}
\bar{g}(m)=\varepsilon \times g(s^{0},E)+\left( 1-\varepsilon \right) \times 
\frac{1}{K}\dsum\limits_{k=2}^{K+1}\rho ^{k}(m)
\end{equation*}%
\noindent where $\varepsilon >0$ is chosen to be small, and $g(\cdot )$ is
the outcome function of the base (unaugmented) mechanism.

\noindent That is, the outcome is a lottery combining the outcome of the
base mechanism using the zeroth and evidence reports, and an outcome
function defined for each round which is the outcome corresponding to a
state on which at least $I-1$ agents have agreed or (if there is no such
agreement), some random lottery \thinspace $b$.

\subsubsection{Transfers}

The mechanism multiplies all the transfers of the baseline mechanism by $%
\varepsilon $ and adds the following transfers.

\begin{equation*}
\tau _{i}^{5}\left( m\right) =\left\{ 
\begin{tabular}{l}
$-\alpha $ \\ 
$0$%
\end{tabular}%
\begin{tabular}{l}
if $s_{i}^{1}\nsim s_{i+1}^{0};$ \\ 
otherwise.%
\end{tabular}%
\right.
\end{equation*}

\noindent That is, agent $i$ receives a fine of $\alpha $ if their first
report is not identical to agent $i+1$'s zeroth report.

\begin{equation*}
\tau _{i}^{6}\left( m\right) =\left\{ 
\begin{tabular}{l}
$-\beta $ \\ 
$0$%
\end{tabular}%
\begin{tabular}{l}
if $\exists s\text{ s.t. }s=s_{i}^{1}\text{ }\forall i$ and $s_{i}^{k}\neq s$
and $s_{n}^{m}=s$ $\forall m,n<k;$ \\ 
otherwise.%
\end{tabular}%
\right.
\end{equation*}

\noindent That is, agent $i$ receives a fine of $\beta $ if their $k^{\text{%
th}}$ report is the first deviation from a unanimous first report.

\begin{equation*}
\tau _{i}^{7,k}\left( m\right) =\left\{ 
\begin{tabular}{l}
$-\gamma $ \\ 
$0$%
\end{tabular}%
\begin{tabular}{l}
if $\exists s$ s.t. $s_{i}^{k}\neq s$ and $s_{j}^{k}=s$ $\forall j\neq i;$
\\ 
otherwise.%
\end{tabular}%
\right.
\end{equation*}

\noindent That is, agent $i$ receives a fine of $\gamma $ if their $k^{\text{%
th}}$ report is the only deviation in round $k\,$from an otherwise unanimous
set of reports. This fine is applied to each round.

\subsubsection{Proof of Implementation}

In the following proof, we do not directly prove implementation for any
values of $\alpha ,\beta ,\gamma $ and $K$. Rather, we show that there exist
values of these parameters such that the mechanism implements and the
overall transfer to any agent can be bounded below an arbitrarily small
number which is greater than zero.

\begin{claim}
\label{c1st} In any equilibrium, for every agent $i$, $s_{i}^{0}$ is the
truth.
\end{claim}

\begin{proof}
It is clear that any agent's reports with index $0$ and the evidence
messages only affect their payoffs through the outcome. While all transfers
in the base mechanism are scaled by $\varepsilon $, the maximum utility
value of \ manipulating the outcome is also scaled by $\varepsilon $ (owing
to the randomization in the final outcome), so that in this mechanism, the
reports with index $0$ are indeed the truth.
\end{proof}

\begin{claim}
\label{c2st} There exist values of $\alpha ,\beta $ and $\gamma $ such that
in any equilibrium, for every agent $i$, $s_{i}^{1}$ is the truth.
\end{claim}

\begin{proof}
Suppose not. That is, $\exists i$ s.t. $s_{i}^{1}\neq s^{\ast }$. From Claim %
\ref{c1st}, we know that $s_{i}^{0}=s^{\ast }$ $\forall i$. Then, consider a
deviation $s_{i}^{1}\leftarrow s^{\ast }$. Notice that $s_{i}^{1}$ does not
affect the outcome. Then, the agent gains at least $\alpha $ from $\tau ^{5}$%
, and could lose up to $\beta $ to $\tau ^{6}$. There is no effect from $%
\tau ^{7}$. If $\alpha >\beta ,$ then there is a profitable deviation.
\end{proof}

\begin{claim}
There exists values of $\alpha ,\beta ,\gamma $ and $K$ such that in any
equilibrium, for every agent $i$, $s_{i}^{k}$, $k=2,...K+1$ is the truth.
\end{claim}

\begin{proof}
The proof proceeds by induction. First we prove that $s_{i}^{2}=s^{\ast
}~\forall i$. Suppose not. Then, there is an agent $i$ for whom $%
s_{i}^{2}\neq s^{\ast }$ and $s_{j}^{2}=s^{\ast }$ $\forall j<i$. That is,
agent $i$ is the first deviant from the unanimous report of $s^{\ast }$ in $%
s^{1}$. Now, there are two cases:

\noindent \textbf{Case 1: }$\exists j\neq i$ s.t. $s_{j}^{2}\neq s^{\ast }$.
Then, agent $i$ is the first deviant, but there are other deviants in $s^{2}$%
. Consider a deviation to $s_{i}^{2}\leftarrow s^{\ast }$. This deviation
could cause him a loss of up to $\frac{1}{K}$ from the outcome, but yields a
profit of $\beta $ due to $\tau ^{6}$. In case the agents had unanimously
agreed on a state in $s^{2}$, this could also lead to a loss of at most $%
\gamma $ from $\tau ^{7}$. Thus, as long as $\beta >\frac{1}{K}+\gamma $,
this is a profitable deviation.

\noindent \textbf{Case 2: }$\NEG{\exists}j\neq i$ s.t. $s_{j}^{2}\neq
s^{\ast }$. That is, agent $i$ is the first and only deviant in $s^{2}$.
Consider again the deviation $s_{i}^{2}\leftarrow s^{\ast }$. There is no
change to the outcome (as agent $i$ was the only agent who was disagreeing
with the otherwise unanimous true report), and even though agent $i$ could
be the first deviant in further rounds, so that this deviation does not
necessarily yield any advantage from $\tau ^{6}$, it does yield a profit of $%
\gamma $ owing to $\tau ^{7}$. Thus, we need $\gamma >0$ to make this a
profitable deviation.

Thus, $s_{i}^{2}=s^{\ast }\forall i$. Clearly then, this argument can be
used inductively for the following rounds to establish that $%
s_{i}^{k}=s^{\ast }\forall i,k$.
\end{proof}

\begin{claim}
For any $\bar{\delta}>0$, the overall transfer to an agent can be bounded
below $\bar{\delta}$.
\end{claim}

\begin{proof}
Consider any $0<\delta <\bar{\delta}$. From the previous claims, the
inequalities required to be satisfied are $\gamma >0$, $\beta >\frac{1}{K}%
+\gamma $ and $\alpha >\beta $, while the largest possible transfer is $%
\alpha +\beta +K\gamma $. Consider the following choices: $\gamma =\frac{%
\delta }{3K}$, $\beta =\frac{1}{K}+\frac{\bar{\delta}}{3K}$ and $\alpha =%
\frac{\delta }{3}$. It is clear that $K$ can be chosen to satisfy the
required inequalities, and $\alpha +\beta +K\gamma <\bar{\delta}$.
\end{proof}

\subsection{\label{app:hardcostly} Proof of Theorem \protect\ref{hardcostly}}

The proof is by construction of an implementing mechanism. Below, we present
the mechanism and a formal proof of implementation. But first, we formally
state the following property of the cost function%
\begin{equation*}
\max_{s\in S}\max_{i\in \mathcal{I}}\max_{E_{i}\in \mathcal{E}%
_{i}}c_{i}(E_{i},s)<C
\end{equation*}

which states that across all states, all agents and all articles of
evidence, the cost of evidence is bounded by $C$, a positive number.

\subsubsection{\textbf{Message Space and Outcome}}

The message space and outcome function remain unaltered from the original
implementing mechanism.

\subsubsection{\textbf{Transfers}}

Structurally, the transfers are similar (with the addition of a transfer
penalizing disagreement), however, the amounts are not fixed a priori,
rather we show the existence of transfers so that the mechanism implements.
We have then, the following transfers:

\begin{equation*}
\tau _{ij}^{1}\left( m\right) =\left\{ 
\begin{tabular}{ll}
$-T_{1}$, & if $s_{i}\in E_{j}$ and $s_{j}\not\in E_{i}$; \\ 
$T_{1}$, & if $s_{i}\not\in E_{j}$ and $s_{j}\in E_{i}$; \\ 
$0$, & otherwise.%
\end{tabular}%
\right.
\end{equation*}

\begin{equation*}
\tau _{i}^{2}\left( m\right) =\left\{ 
\begin{tabular}{ll}
$-T_{2}$, & if $\exists j\in \mathcal{I}$ s.t. $E_{j}\not\subseteq
E_{j}^{\ast }(s_{i})$; \\ 
$0\,$, & otherwise.%
\end{tabular}%
\right.
\end{equation*}

\begin{equation*}
\tau _{ij}^{3}\left( m\right) =\left\{ 
\begin{tabular}{ll}
$-T_{3}$, & if $E_{i}^{\ast }(s_{i})\not=E_{i}^{\ast }(s_{j})$; \\ 
$0$, & $\text{otherwise.}$%
\end{tabular}%
\right.
\end{equation*}

\begin{equation*}
\tau _{i}^{4}\left( m\right) =\left\{ 
\begin{tabular}{ll}
$-\frac{T_{4}}{|S|}|E_{i}|$, & if $E_{j}\not\subseteq E_{j}^{\ast
}(s_{j^{\prime }})$ for some $j,j^{\prime }\in \mathcal{I}$; \\ 
$0$, & otherwise.%
\end{tabular}%
\right.
\end{equation*}

The overall transfer to agent $i$ is given (as before) by:

\begin{equation*}
\tau ^{i}=\sum_{j\neq i}\tau _{ij}^{1}+\tau _{i}^{2}+\sum_{j\neq i}\tau
_{ij}^{3}+\tau _{i}^{4}
\end{equation*}

\subsubsection{Proof of Implementation}

In what follows, we assume that the true state is $s^{\ast }$. We pick $%
T_{1} $, $T_{2}$, $T_{3}$, and $T_{4}$ so that the following inequalities
are satisfied:%
\begin{equation}
T_{1}>\frac{C|S|}{\varepsilon }  \label{A}
\end{equation}%
\begin{equation}
T_{1}\geq 1+T_{2}+(I-1)T_{3}+T_{4}  \label{C}
\end{equation}

\begin{equation}
T_{2}\geq 1+(I-1)T_{3}  \label{H}
\end{equation}%
\begin{equation}
T_{3}\geq 1  \label{J}
\end{equation}%
\begin{equation}
T_{4}>\frac{C|S|}{1-\varepsilon }  \label{G}
\end{equation}

It is immediate that this system of inequalities has a feasible solution.
For instance, first choose $\varepsilon \in (0,1)$, followed by setting $%
T_{4}=\frac{C|S|}{1-\varepsilon }+\varepsilon $. Clearly, Inequality (\ref{G}%
) is satisfied. Second, set $T_{3}$, $T_{2}$, and $T_{1}$ in order so that
Inequalities (\ref{J}), (\ref{H}), and (\ref{C}) are satisfied. Before we
proceed to present the proof, we first provide a sketch to outline the
augments.

First, we can prevent the agents from reporting other-refutable lies with a
probability $\varepsilon $ or more, as the reward for refutation, $T_{1}$,
can be made large enough that doing so with probability $\varepsilon $ (or
more) guarantees refutation by other agents. This guaranteed by Inequality (%
\ref{C}) and established in Claims \ref{cc_ORL_1}--\ref{cc_ORL_2}. With
this, we are able to establish that all agents present their tightest
evidence (Claim \ref{cc_SRL_1}). This is achieved by choosing $T_{4}$ high
enough so that the necessity of supporting the other agent's state claims
overrides the evidence cost. Notice that owing to Observation 1, every state
report except an other-refutable lie requires the presentation of tightest
evidence to support. Second, agents do not present non-refutable lies,
because of the penalty from $\tau ^{2}$, which is easily avoided by
deviating to the truth (since all evidence is the tightest). This is
outlined in Claims \ref{cc_NRL_2} and \ref{cc_NRL_3}. Third, we establish
that any self-refutable lies which are presented must be consistent with the
tightest evidence of all other agents (Claim \ref{cc_SRL_2}) and thereby
eliminate the possibility that any self-refutable lies are presented in
equilibrium (Claim \ref{cc_SRL_3}). \ From the preceding logic, it is clear
that if agents present their tightest evidence, then no agent presents
other-refutable lies, and thus implementation obtains.

We now present below a formal proof of implementation using the mechanism
stated above. In what follows, we denote the true state by $s^{\ast }$. Fix
an arbitrary mixed-strategy Nash equilibrium $\sigma $.

\paragraph{Bounding the Probability of Other-Refutable Lies}

\begin{claim}
\label{cc_ORL_1}If agent $i$ reports with probability at least $\frac{%
\varepsilon }{|S|}$ a state claim $s_{i}$ such that agent $j\neq i$ has an
article of evidence which refutes $s_{i}$, then agent $j$ must refute $s_{i}$
with probability $1$, i.e., $E_{j}$ refutes $s_{i}$ for every $%
m_{j}=(s_{j},E_{j})$ on the support of $\sigma _{j}$.
\end{claim}

\begin{proof}
Presenting this evidence nets agent $j$ a reward of at least $\frac{%
\varepsilon }{|S|}T_{1}$, and costs at most $C$. Thus, if $T_{1}>\frac{C|S|}{%
\varepsilon }$, then this is a profitable deviation. This is a consequence
of Inequality (\ref{A}).
\end{proof}

\begin{claim}
\label{cc_ORL_2}Each agent reports other-refutable lies with a total
probability less than $\varepsilon $.
\end{claim}

\begin{proof}
Suppose not. Then, there is an agent $i$ who reports a message $m_{i}$ with
a probability $\frac{\varepsilon }{|S|}$ or more such that agent $j\neq i$
has an article of evidence $E_{j}$ in $\mathcal{E}_{j}(s^{\ast })$ which
refutes $s_{i}$. From Claim \ref{cc_ORL_1}, agent $j$ presents $E_{j}$ with
probability 1. Then, the following table shows agent $i$'s payoff changes
from switching to the truth:

\begin{equation*}
\begin{tabular}{|l|l|l|l|l|}
\hline
$g$ & $\tau _{i}^{1}$ & $\tau _{i}^{2}$ & $\tau _{i}^{3}$ & $\tau _{i}^{4}$
\\ \hline
$>-1$ & $\geq T_{1}$ & $\geq -T_{2}$ & $\geq -(I-1)T_{3}$ & $\geq -T_{4}$ \\ 
\hline
\end{tabular}%
\end{equation*}

\noindent Since $\varepsilon <1$, it follows from Inequality (\ref{C}) that $%
T_{1}\geq 1+T_{2}+(I-1)T_{3}+T_{4}$. Hence, this constitutes a profitable
deviation.
\end{proof}

\begin{claim}
\label{cc_ORL_3}If all other agents present the tightest evidence with
probability one, then no agent reports an other-refutable lie.
\end{claim}

\begin{proof}
Suppose not. Then, there is an agent $i$ who reports a message $m_{i}$ such
that agent $j\neq i$ has an article of evidence $E_{j}$ in $\mathcal{E}%
_{j}(s^{\ast })$ which refutes $s_{i}$, while agent $j$ presents this
article of evidence with probability one. Then, the following table shows
agent $i$'s payoff changes from switching to the truth:

\begin{equation*}
\begin{tabular}{|l|l|l|l|l|}
\hline
$g$ & $\tau _{i}^{1}$ & $\tau _{i}^{2}$ & $\tau _{i}^{3}$ & $\tau _{i}^{4}$
\\ \hline
$>-1$ & $\geq T_{1}$ & $\geq -T_{2}$ & $\geq -(I-1)T_{3}$ & $\geq -T_{4}$ \\ 
\hline
\end{tabular}%
\end{equation*}

\noindent Inequality (\ref{C}) implies that this is a profitable deviation.
\end{proof}

\paragraph{Bounding the Probability of Nonrefutable Lies}

\begin{claim}
\label{cc_SRL_1}Every agent presents their tightest evidence with
probability one.
\end{claim}

\begin{proof}
Consider an arbitrary agent $i$. From Claim \ref{cc_ORL_2}, other agents are
presenting state claims which agent $i$ cannot refute with a probability $%
1-\varepsilon $ or more. From Observation 1, these claims require agent $i$
to present his tightest evidence so that they can be supported. Therefore,
on any message where agent $i$ does not present his tightest evidence, he
expects $\tau ^{4}$ to be active with probability $1-\varepsilon $ or more.
Since $\frac{(1-\varepsilon )T_{4}}{|S|}>C$ (Inequality (\ref{G})), it is a
profitable deviation to present the tightest evidence on such messages.
\end{proof}

\begin{claim}
\label{cc_NRL_2}If all agents present their tightest evidence, then no agent
reports nonrefutable lies.
\end{claim}

\begin{proof}
Suppose not. Then, there is an agent $i$ who reports a message $m_{i}$ which
contains a nonrefutable lie while all agents present their tightest
evidence. In any message $m_{i}=(s_{i},E_{i}^{\ast }(s^{\ast }))$ where $%
s_{i}$ is a nonrefutable lie, consider a deviation to the truth. Then, the
following table shows agent $i$'s payoff changes from switching to the truth:%
\begin{equation*}
\begin{tabular}{|l|l|l|l|l|}
\hline
$g$ & $\tau _{i}^{1}$ & $\tau _{i}^{2}$ & $\tau _{i}^{3}$ & $\tau _{i}^{4}$
\\ \hline
$>-1$ & $0$ & $T_{2}$ & $\geq -(I-1)T_{3}$ & $\geq 0$ \\ \hline
\end{tabular}%
\end{equation*}

First, the agent loses less than $1$ from changing the outcome, incurs no
loss from $\tau ^{1}$ (as the truth is irrefutable) and at most $(I-1)T_{3}$
from $\tau ^{3}$. The agent incurs no losses from $\tau ^{4}$ either as $%
s_{i}$ was unsupported but the truth is supported and the size of his
evidence set has not changed. Since all the presented evidence is the
tightest, it follows that the agent gains at least $T_{2}$ from $\tau ^{2}$
from the deviation (as the truth is supported by all agents). It follows
from Inequality (\ref{H}) that this is a profitable deviation.
\end{proof}

\begin{claim}
\label{cc_NRL_3}No agent reports nonrefutable lies.
\end{claim}

\begin{proof}
This follows immediately from Claims \ref{cc_SRL_1} and \ref{cc_NRL_2}.
\end{proof}

\paragraph{Bounding the Probability of Self-Refutable Lies}

\begin{claim}
\label{cc_SRL_2}For any agent $i$, a self-refutable lie $s_{i}$ is reported
with positive probability only if $E_{j}^{\ast }(s_{i})=E_{j}^{\ast
}(s^{\ast })$ for every agent $j\neq i$.
\end{claim}

\begin{proof}
Suppose to the contrary that agent $i$ reports $s_{i}\in SRL_{i}$ such that $%
E_{j}^{\ast }(s_{i})\not=E_{j}^{\ast }(s^{\ast })$ for some $j$. Then, by
Observation 1, $E_{j}^{\ast }(s_{i})\subset E_{j}^{\ast }(s^{\ast })$.
Consider then an alternative message which replaces $s_{i}$ with the truth.
The following table shows agent $i$'s payoff changes:%
\begin{equation*}
\begin{tabular}{|l|l|l|l|l|}
\hline
$g$ & $\tau _{i}^{1}$ & $\tau _{i}^{2}$ & $\tau _{i}^{3}$ & $\tau _{i}^{4}$
\\ \hline
$>-1$ & $0$ & $\geq T_{2}$ & $\geq -(I-1)T_{3}$ & $\geq 0$ \\ \hline
\end{tabular}%
\end{equation*}

First, the agent loses less than 1 to the outcome. Second, the truth is not
refutable so that the agent incurs no losses to $\tau ^{1}$. Third, we
establish that the agent gains $T_{2}$ on account of $\tau ^{2}$. From Claim %
\ref{cc_SRL_1}, all agents present the tightest evidence with probability
one. Then, the truth incurs no penalty while $s_{i}$ incurs a penalty since $%
E_{j}^{\ast }(s_{i})\subset E_{j}^{\ast }(s^{\ast })$ and therefore $s_{i}$
was not supported. Fourth, the agent loses at most $(I-1)T_{3}$ to $\tau
^{3} $, as other agents may not be reporting state reports consistent with
the truth in $i$'s evidence, i.e., it is possible that $E_{i}^{\ast
}(s_{j})\not=E_{i}^{\ast }(s^{\ast })$). Fifth, since all other agents
report the tightest evidence, the truth is supported, so that there is no
loss from $\tau ^{4}$. Inequality (\ref{H}) then implies that this is a
profitable deviation.
\end{proof}

\begin{claim}
\label{cc_SRL_3}Agents do not report self-refutable lies.
\end{claim}

\begin{proof}
Suppose not. That is, suppose that there is an agent $i$ who reports a
self-refutable lie $s_{i}$. Consider a deviation to the truth for agent $i$
(from Claim \ref{cc_SRL_1}, evidence presentation was already the tightest).
Then, the following table shows agent $i$'s payoff changes from this
deviation:%
\begin{equation*}
\begin{tabular}{|l|l|l|l|l|}
\hline
$g$ & $\tau _{i}^{1}$ & $\tau _{i}^{2}$ & $\tau _{i}^{3}$ & $\tau _{i}^{4}$
\\ \hline
$>-1$ & $0$ & $\geq 0$ & $\geq (I-1)T_{3}$ & $\geq 0$ \\ \hline
\end{tabular}%
\end{equation*}

First, the agent loses less than 1 to the outcome. Second, the truth is not
refutable so that the agent incurs no losses to $\tau ^{1}$. Third, the
deviation causes no loss from $\tau ^{2}$ since every agent presents their
tightest evidence and the truth is supported. Fourth, since every agent
presents their tightest evidence, no other-refutable lies and nonrefutable
lies are presented (Claims \ref{cc_ORL_3} and \ref{cc_NRL_2}). Then, a
self-refutable lie disagrees with all the reports (the truth and
self-refutable lies) of other agents in agent $i\,$'s evidence (Observation
1) and hence incurs a penalty from $\tau ^{3}$. The truth avoids this fine
against every message and thus yields a profit of $(I-1)T_{3}$. Fifth, the
truth is supported, so that $\tau ^{4}$ can not lead to any losses either.
Inequality (\ref{J}) yields that this is a profitable deviation.
\end{proof}

Thus, the only equilibria are such that all agents report the truth with
probability $1-\varepsilon $ or more, where $0<\varepsilon <1$ as chosen
earlier. We conclude the proof in the following claim.

\begin{claim}
All agents report the true state with probability 1 and there is no transfer
on the equilibrium.
\end{claim}

\begin{proof}
By Claim \ref{cc_SRL_1}, all agents present the tightest evidence. Hence, it
follows from Claim \ref{cc_ORL_3} that other-refutable lies are not
reported, and from Claim \ref{cc_NRL_2} that nonrefutable lies are not
presented, and Claim \ref{cc_SRL_3} that self-refutable lies are not
presented. Therefore, all agents report the true state with probability 1.
Since the tightest evidence also supports their true state claim, there is
no transfer in equilibrium.
\end{proof}

\subsection{\label{app:softcostly}Proof of Theorem \protect\ref{main_soft}}

\subsubsection{\label{Soft_Mech}Implementing Mechanism}

The message space of the mechanism is given by $M=\Pi _{i}M_{i}$, where $%
M_{i}=S\times \mathcal{E}_{i}$. That is, this is a direct mechanism. A
typical message for agent $i$ is represented as $m_{i}=(s_{i},E_{i})$.

An agent $i$ challenges a state claim $s$ when he presents a message $%
(s_{i},E_{i}^{\gamma })$ if $\exists t$ such that $c_{i}(E_{i}^{\ast
}(s),s)\leq c_{i}(E_{i}^{\gamma },s)-t$ but $c_{i}(E_{i}^{\ast
}(s),s_{i})>c_{i}(E_{i}^{\gamma },s_{i})-t$. When agent $i$ challenges $s$
at $s^{\prime }$ using $E_{i}^{\gamma }(s,s^{\prime })$, the profit he makes
(under constant preferences) is given by $t-c_{i}(E_{i}^{\gamma
}(s,s^{\prime }),s^{\prime })+c_{i}(E_{i}^{\ast }(s),s^{\prime })$. We
pre-select a value $t_{i}(s,s^{\prime })$ for any pair of states $s$ and $%
s^{\prime }$ such that agent $i$ can challenge $s$ at $s^{\prime }$ so that $%
c_{i}(E_{i}^{\ast }(s),s)\leq c_{i}(E_{i}^{\gamma }(s,s^{\prime
}),s)-t_{i}(s,s^{\prime })$ but $c_{i}(E_{i}^{\ast }(s),s^{\prime
})>c_{i}(E_{i}^{\gamma }(s,s^{\prime }),s^{\prime })-t_{i}(s,s^{\prime })$.
Observe that $t_{i}(s,s^{\prime })\in \left( -1,1\right) $ since $\left\vert
c_{i}(E_{i},s)-c_{i}(E_{i},s^{\prime })\right\vert <1$ for any $i$, $E_{i}$, 
$s$ and $s^{\prime }$. In this case, we write $\Gamma _{i}(s,s_{i})=1$ to
denote that agent $i$ has challenged $s$ at the state $s_{i}$.

The outcome is determined as follows:

\begin{equation*}
g\left( m\right) =\left\{ 
\begin{tabular}{ll}
$f(s_{1})$, & if $\Gamma _{i}(s_{1},s_{i})=1$ for $i\neq 1$; \\ 
$f(s)$, & if $\forall i\neq 1$, $s_{i}=s$, $\Gamma _{i}(s_{1},s)=0$ and $%
\Gamma _{1}(s,s_{1})=1$; \\ 
$f(s_{1})$, & otherwise.%
\end{tabular}%
\right.
\end{equation*}

The mechanism employs the following transfers. The first transfer penalizes
agent $1$ an amount of $1$ dollar for every agent who makes a valid
challenge to $s_{1}$. That is,%
\begin{equation*}
\tau _{1,i}^{1}(m)=-1\text{ if }\Gamma _{i}(s_{1},s_{i})=1.
\end{equation*}

The second transfer, which applies only to agent 1, incentivizes him to
agree (in the state dimension) with a report of $(s,E_{i}^{\ast }(s))_{i\neq
1}$ if such a report is made unless he issues a challenge. Formally,

\begin{equation*}
\tau _{1}^{2}(m)=\left\{ 
\begin{tabular}{ll}
$-1$, & if $s_{i}=s$ $\forall i\neq 1$ and $m_{1}\neq (s,E_{1}^{\ast }(s))$
and $\Gamma _{1}(s,s_{1})=0$; \\ 
$0$, & otherwise.%
\end{tabular}%
\right.
\end{equation*}

The third transfer applies only to agents other than 1. It incentivizes them
to either (i) agree with the agent with the least index challenging agent 1
along the state dimension or (ii) agree with agent 1 if no agent is
challenging agent 1.

\begin{equation*}
\tau _{i}^{3}(m)=\left\{ 
\begin{tabular}{ll}
$-1$, & if $s_{i}\neq s_{j}$ where $j=\min \{k:\Gamma _{k}(s_{1},s_{k})=1\}$;
\\ 
$-1$ & if $\{k:\Gamma _{k}(s_{1},s_{k})=1\}=\varnothing $ and $m_{i}\neq
(s_{1},E_{i}^{\ast }(s_{1}))$ \\ 
$0$, & otherwise.%
\end{tabular}%
\right.
\end{equation*}

The fourth transfer is related to the challenge payouts. For agent 1, we have

\begin{equation*}
\tau _{1}^{4}\left( m\right) =\left\{ 
\begin{tabular}{ll}
$t_{1}(s,s_{1})$, & if $\forall i\neq 1$, $s_{i}=s,\Gamma
_{i}(s_{1},s)=0,\Gamma _{1}(s,s_{1})=1$; \\ 
$0$, & otherwise.%
\end{tabular}%
\right.
\end{equation*}

For agent $i\neq 1$, we have

\begin{equation*}
\tau _{i}^{4}\left( m\right) =\left\{ 
\begin{tabular}{ll}
$t_{i}(s_{1},s_{i})$, & if $\Gamma _{i}(s_{1},s_{i})=1$; \\ 
$0$, & otherwise.%
\end{tabular}%
\right.
\end{equation*}

Intuitively, the mechanism works as follows. First, we prevent agent 1 from
presenting any lies that other agents can challenge by encouraging others to
challenge when possible ($\tau ^{4}$ provides this incentive) and penalizing
him for each challenge against his claim ($\tau ^{1}$ yields this penalty).
If the other agents are telling the truth, $\tau ^{2}$ provides agent 1 the
incentive to agree with them. If agent 1 is challenged, $\tau ^{3}$ ensures
that every challenge is mounted with the same state claim. Then, agent 1 can
deviate to match this common state avoiding all the penalties from $\tau
^{1} $. This forms a profitable deviation. This leaves us with the truth,
and lies that only agent 1 can challenge. If agent 1 tells a lie that only
he can challenge, we get the other agents to agree with him using $\tau ^{3}$
(this helps the designer figure out what state is being challenged) and
allow agent 1 to challenge this agreement. Note that other agents are not
penalized when agent 1 challenges.

\subsubsection{Proof of Implementation}

The following lemma, which is a property of the evidence structure, finds
use later.

\begin{lemma}
\label{no_chal_cheapest}If $s^{\prime }$ is a lie that only agent $i$ can
challenge at $s^{\ast }$, then agent $j$ ($\neq i$) cannot challenge $%
s^{\ast }$ at $s^{\prime }$ with evidence $E_{j}^{\ast }(s^{\prime })$.
\end{lemma}

\begin{proof}
$s^{\prime }$ is a lie that only $i$ can challenge at $s^{\ast }$. Then, $j$
cannot challenge $s^{\prime }$ at $s^{\ast }$. Therefore, $\forall E\in 
\mathcal{E}_{j}$, $c_{j}(E,s^{\ast })-c_{j}(E_{j}^{\ast }(s^{\prime
}),s^{\ast })\geq c_{j}(E,s^{\prime })-c_{j}(E_{j}^{\ast }(s^{\prime
}),s^{\prime })$. With $E=E_{j}^{\ast }(s^{\ast })$, this yields $%
c_{j}(E_{j}^{\ast }(s^{\ast }),s^{\ast })-c_{j}(E_{j}^{\ast }(s^{\prime
}),s^{\ast })\geq c_{j}(E_{j}^{\ast }(s^{\ast }),s^{\prime
})-c_{j}(E_{j}^{\ast }(s^{\prime }),s^{\prime })$. If $j$ can challenge $%
s^{\ast }$ at $s^{\prime }$ with challenge evidence $E_{j}^{\ast }(s^{\prime
})$, then we must have $c_{j}(E_{j}^{\ast }(s^{\prime }),s^{\prime
})-c_{j}(E_{j}^{\ast }(s^{\ast }),s^{\prime })<c_{j}(E_{j}^{\ast }(s^{\prime
}),s^{\ast })-c_{j}(E_{j}^{\ast }(s^{\ast }),s^{\ast })$, which is a
contradiction.
\end{proof}

Essentially, Lemma \ref{no_chal_cheapest} allows the designer to deduce that
in a profile of the form $((s_{1},E_{1}),(s,E_{2}^{\ast }(s)),(s,E_{3}^{\ast
}(s)),...,(s,E_{I}^{\ast }(s)))$, if agent $1$ cannot challenge $s$ at $%
s_{1} $, then it is actually agent $1$ challenging $s$ rather than agents $%
i\neq 1$ challenging $s_{1}$. In a canonical mechanism, the need for three
agents is often to identify the agent who is deviating from the majority.
Lemma \ref{no_chal_cheapest} allows us to dispense with this requirement in
the costly evidence setting (under constant preferences).

Suppose the true state is $s^{\ast }$ and consider any pure strategy
equilibrium $((s_{i},E_{i}))_{i\in \mathcal{I}}$.

\begin{claim}
\label{2_Ch}If there is an agent $i\neq 1$ who can challenge $s_{1}$, then $%
s_{1}$ is challenged.
\end{claim}

\begin{proof}
Suppose not. Then, no agent challenges $s_{1}$. Consider a deviation for $i$
to $(s_{i},E_{i}^{\gamma })$ which challenges $s_{1}$. The outcome remains $%
f(s_{1})$ and agent $i$ does not incur any penalties from $\tau ^{3}$ as he
is the first agent who challenges $s_{1}$. He gains a reward from $\tau ^{4}$%
, so that this is a profitable deviation.
\end{proof}

\begin{claim}
\label{2_Ch_cons}If there is an agent $i\neq 1$ who can challenge $s_{1}$,
then $\exists s\in S$ such that $s_{i}=s$, $\forall i\neq 1$.
\end{claim}

\begin{proof}
If there are only two agents, then this claim is trivially satisfied. So
suppose there are three or more agents. From Claim \ref{2_Ch}, $s_{1}$ is
challenged by some agent. Suppose $i$ is the first agent to challenge $s_{1}$%
. For any agent $j\neq i$, $j\neq 1$, if $s_{j}\neq s_{i}$, then he gets a
penalty of $1$ dollar from $\tau ^{3}$. Consider a deviation to match $s_{i}$%
. The outcome remains $f(s_{1})$, but the agent avoids the penalty from $%
\tau ^{3}$. Since $t_{i}(s,s^{\prime })\in \left( -1,1\right) $, any
possible reward from challenging using $s_{j}$ (from $\tau ^{4}$) is less
than $1$ dollar. Hence, this is a profitable deviation.
\end{proof}

\begin{claim}
\label{1_noC2}Agent 1 does not present lies which others can challenge.
\end{claim}

\begin{proof}
From Claim \ref{2_Ch}, if agent 1 presents such a lie, then some other agent
challenges it. Further, from Claim \ref{2_Ch_cons}, $\exists s\in S$ such
that $s_{i}=s,\forall i\neq 1$. This yields agent 1 a penalty of at least 1
dollar from $\tau ^{1}$, which he can avoid by deviating to present $%
(s,E_{1}^{\ast }(s))$. This may change the outcome, but the utility loss
from that is less than 1 dollar, so that this forms a profitable deviation.
\end{proof}

\begin{claim}
\label{Agree}If agent 1 presents a lie that only he can challenge at $%
s^{\ast }$, then every other agent agrees with him in the state and evidence
dimensions.
\end{claim}

\begin{proof}
We note that if $s_{1}$ is a lie that agent $i\neq 1$ cannot challenge at $%
s^{\ast }$, then it is among his best responses to submit $E_{i}^{\ast
}(s_{1})$ and obtain the outcome $f(s_{1})$. If he does not present $%
(s_{1},E_{i}^{\ast }(s_{1}))$, then agent $i$ is either mounting a challenge
which yields him a loss (as otherwise $s_{1}$ would be a lie agent $i$ could
challenge), or disagreeing without challenging and incurring a loss of $1$
dollar. In either case, it is best for agent $i$ to deviate to match agent
1's claim in both dimensions.
\end{proof}

\begin{claim}
\label{1_no_SCL}Agent 1 does not present a lie that only he can challenge at 
$s^{\ast }$.
\end{claim}

\begin{proof}
Suppose not. From Claim \ref{Agree}, if $s_{1}$ is a lie that only agent 1
can challenge at $s^{\ast }$, then every agents $i\neq 1$ presents $%
(s_{1},E_{i}^{\ast }(s_{1}))$ and the outcome is $f(s_{1})$. Disagreeing
without challenging is suboptimal owing to $\tau ^{2}$. Therefore, consider
a deviation for agent 1 to $(s^{\ast },E_{1}^{\gamma }(s_{1},s^{\ast }))$.
From Lemma \ref{no_chal_cheapest}, when agents $i\neq 1$ present $%
(s_{1},E_{i}^{\ast }(s_{1}))$, $\Gamma _{i}(s^{\ast },s_{1})=0$. However, $%
\Gamma _{1}(s_{1},s^{\ast })=1$. This yields him a profit from challenging
since the outcome remains $f(s_{1})$.
\end{proof}

\begin{claim}
The mechanism implements.
\end{claim}

\begin{proof}
Owing to Claims \ref{1_noC2} and \ref{1_no_SCL}, agent 1 can only present
claims which no one can challenge. In such case, it is optimal owing to $%
\tau _{1}^{2}$ and $\tau ^{3}$ for all agents to also present the designated
cheapest evidence along with such claims. This leads to no transfers, the
correct $f$-optimal outcome, and the submission of only the cheapest
evidence in equilibrium. This satisfies the definition of implementation.
\end{proof}

\subsubsection{\label{app:meas_but_not_em} Measurability does not imply
Evidence Monotonicity}

Consider the following evidence structure.

\begin{equation*}
\begin{tabular}{|c|c|c|}
\hline
State/Agent & $1$ & $2$ \\ \hline
$s_{1}$ & $\{\{s_{1},s_{2},s_{3},s_{4}\}\}$ & $\{\{s_{1},s_{2},s_{3},s_{4}\}%
\}$ \\ \hline
$s_{2}$ & $\{\{s_{1},s_{2},s_{3},s_{4}\},\{s_{2},s_{4}\}\}$ & $%
\{\{s_{1},s_{2},s_{3},s_{4}\}\}$ \\ \hline
$s_{3}$ & $\{\{s_{1},s_{2},s_{3},s_{4}\},\{s_{3},s_{4}\}\}$ & $%
\{\{s_{1},s_{2},s_{3},s_{4}\}\}$ \\ \hline
$s_{4}$ & $\{\{s_{1},s_{2},s_{3},s_{4}\},\{s_{2},s_{4}\},\{s_{3},s_{4}\},%
\{s_{4}\}\}$ & $\{\{s_{1},s_{2},s_{3},s_{4}\}\}$ \\ \hline
\end{tabular}%
\end{equation*}

The cost of the article $\{s_{4}\}$ is positive but finite.

Clearly, any $f$ is measurable with respect to the evidence structure since
between any pair of states, agent 1 has a different endowment. To see that $%
f $ is not evidence monotonic, consider the states $s_{1}$ and $s_{4}$. Note
that of necessity, $E_{1}^{\ast }(s_{1})=\{s_{1},s_{2},s_{3},s_{4}\}$. We
now deal with three cases. First, if $E_{1}^{\ast
}(s_{4})=\{s_{1},s_{2},s_{3},s_{4}\}$, then no agent can challenge $s_{4}$
at $s_{1}$ since relative to $E_{1}^{\ast }(s_{4})$, the cost of all
articles of evidence are strictly higher at $s_{1}$ than at $s_{4}$ for
agent 1, and agent 2 has no cost variation. Recall that for $f$ to be
evidence monotonic under constant preference, it is necessary that the cost
of some article of evidence reduces relative to $E_{1}^{\ast }(s_{4})$ in
going from $s_{4}$ to $s_{1}$. Instead the articles $\{s_{2},s_{4}\}$ and $%
\{s_{3},s_{4}\}$ rise in cost from being costless to becoming unavailable
(costing $\infty $).

Second, if $E_{1}^{\ast }(s_{4})=\{s_{2},s_{4}\}$, then no agent can
challenge $s_{4}$ at $s_{2}$ because as in the above case, the articles $%
\{s_{1},s_{2},s_{3},s_{4}\}$ \& $\{s_{2},s_{4}\}$ continue to be costless,
while the article $\{s_{3},s_{4}\}$ has now become unavailable. Third, if $%
E_{1}^{\ast }(s_{4})=\{s_{3},s_{4}\}$, then no agent can challenge $s_{4}$
at $s_{3}$ since the articles $\{s_{1},s_{2},s_{3},s_{4}\}$ \& $%
\{s_{3},s_{4}\}$ continue to be costless, while the article $\{s_{2},s_{4}\}$
has now become unavailable. Further, $E_{1}^{\ast }(s_{4})\not=\{s_{4}\}$
since $\{s_{4}\}$ is not among the least costly evidence. Therefore, $f$ is
not evidence monotonic under constant preference.

\subsection{\label{CEMS}Costly Evidence and Mixed Strategies}

The setup we study, and the definition of implementation remains the same as
Section 4.2. In addition, we define the following notation. First, define $%
\gamma =\min_{i}\min_{s}\min_{E_{i}\in \mathcal{E}_{i}\backslash \mathcal{E}%
_{i}^{l}}[c_{i}(E_{i},s)-c_{i}^{l}(s)]$. $\gamma $ is the minimum cost
difference between the cheapest evidence and non-cheapest evidence across
all agents and states. Second, define $N=\max_{i}\max_{s}|\mathcal{E}%
_{i}^{l}(s)|$, that is, $N$ is the maximum cardinality of the cheapest
evidence set for any agent in any state.

\subsubsection{Evidence Monotonicity}

If the designer were limited to only making arbitrarily small rewards, then
Definition \ref{emucs} can be reinterpreted as follows.

\begin{definition}
\label{EM*}An SCF $f$ is evidence-monotonic under constant preferences and
arbitrarily small rewards if there exists $E^{\ast }:S\rightarrow \mathcal{E}
$ such that

\begin{enumerate}
\item[(i)] for all $s$, $E^{\ast }(s)\in \mathcal{E}^{l}(s,f(s))$

\item[(ii)] for all $s$ and $s^{\prime }$, if

$\forall i,E_{i}^{\prime }:[c_{i}(E_{i}^{\prime },s)\geq c_{i}(E_{i}^{\ast
}(s),s)\implies c_{i}(E_{i}^{\prime },s^{\prime })\geq c_{i}(E_{i}^{\ast
}(s),s^{\prime })]$,

then $f(s)=f(s^{\prime })$.
\end{enumerate}

Alternatively, if $f(s^{\prime })\neq f(s)$, then there exists $E^{\ast
}:S\rightarrow \mathcal{E}$ such that

\begin{enumerate}
\item[(i)] for all $s$, $E^{\ast }(s)\in \mathcal{E}^{l}(s)$

\item[(ii)] $\exists i,E_{i}^{\prime }:[c_{i}(E_{i}^{\prime },s)\geq
c_{i}(E_{i}^{\ast }(s),s)$ but $c_{i}(E_{i}^{\ast }(s),s^{\prime
})>c_{i}(E_{i}^{\prime },s^{\prime })]$
\end{enumerate}

Then, $E_{i}^{\ast }(s)\notin \mathcal{E}_{i}^{l}(s^{\prime })$
\end{definition}

We refer to this condition as evidence monotonicity*. Under evidence
monotonicity*, if two states have different outcomes, there must be an
article of evidence which was cheapest earlier, but is not cheapest now.
Since there could be multiple articles of evidence which share the same
cheapest cost, there could be many cheapest articles which are common
between the states, but the condition above implies that this common set
cannot include \emph{all} of them.

Under evidence monotonicity*, states with different outcomes have different
sets of cheapest evidence. Therefore, in what follows, we redefine the
notion of a "state" to mean a particular profile of cheapest evidence. Then,
between any two states, there is at least one agent who has a different set
of cheapest evidence.

\subsubsection{Challenges and Classification of Lies}

In the hard evidence setup, we defined a lie as a state which induced a
different tightest evidence for at least one agent relative to the truth. In
this setup, we instead define a lie as a state which induces a different
cheapest evidence set for at least one agent relative to the truth. That is, 
$L(s^{\ast })=\{s\in S:\mathcal{E}_{i}^{l}(s)\neq \mathcal{E}^{l}(s^{\ast })$
for some $i\}$. We note that it is not necessary that a lie must have a
different outcome than the truth, rather, it should be possible to
distinguish a lie from the truth using a mechanism that only allows the
submission of cheapest evidence on path.

As in the hard evidence setup, it will be useful to classify lies into
various categories. Whereas under hard evidence, it was possible to outright
refute lies, it may no longer be possible to do so under a costly evidence
setup, because all the articles of evidence are available to the agent in
all states, irrespective of the state. However, we can define a certain
notion like refutability even under a costly evidence setup, using the set
of cheapest evidence.

\begin{definition}
An agent $i$ challenges a state $s$ by presenting an article of evidence $%
E_{i}$ if $E_{i}\notin \mathcal{E}_{i}^{l}(s)$.
\end{definition}

If we were to design a mechanism in which it is always suboptimal for an
agent to present articles of evidence which are not the cheapest under the
true state, then an agent credibly signals to the designer that the true
state is not $s$ if he presents an article of evidence which is not among
his cheapest under the state $s$.

We note that due to evidence monotonicity*, for each lie, there must be at
least one agent who is able to challenge it. Below, we define the set of
lies challengeable by agent $i$. A lie $s$ is said to be challengeable by $i$
at state $s^{\ast }$ if agent $i$ is among the agents who can challenge it.
We denote the set of such lies as $CL_{i}(s^{\ast })$.

We define the set of \emph{other-challengeable} lies for agent $i$ as the
set of lies which agent $i$ can tell which agents other than $i$ can
challenge. That is, $OCL_{i}(s^{\ast })=\cup _{j\neq i}CL_{j}(s^{\ast })$.
Since we will be interested in the set of such lies under the true state
(usually denoted by $s^{\ast }$), we will often supress the parenthetical,
referring to such lies merely as members of the set $OCL_{i}$.

Similarly, the set of \emph{self-challengeable }lies for agent $i$ are the
lies which only agent $i$ can challenge, that is, $SCL_{i}(s^{\ast
})=L(s^{\ast })\backslash \cup _{j\neq i}CL_{j}(s^{\ast })$.

We prove the following property of self-challengeable lies which finds use
later.

\textbf{Observation 1*}: If $s_{i}\in SCL_{i}(s^{\ast })$, then $\mathcal{E}%
_{j}^{l}(s^{\ast })\subseteq \mathcal{E}_{j}^{l}(s_{i})~\forall j\neq i$.
That is, for any agent, the set of cheapest evidence for the true state is
contained within the set of cheapest evidence induced by another agent's
self-challengeable lie.

\textbf{Proof:} If not, there is an article of evidence $E_{j}\in \mathcal{E}%
_{j}^{l}(s^{\ast })$ such that $E_{j}\not\in \mathcal{E}_{j}^{l}(s_{i})$, in
which case $E_{j}$ challenges $s_{i}$, contradicting the claim that $s_{i}$
is self-challengeable for agent $i$.

We note that Definition \ref{EM*} implies a nature of symmetry. That is, if $%
s$ and $s^{\prime }$ induce different outcomes, then both of the following
are true simultaneously.

\begin{itemize}
\item There is an article of evidence which was cheapest under $s$, but is
not cheapest under $s^{\prime }$.

\item There is an article of evidence which was cheapest under $s^{\prime }$%
, but is not cheapest under $s$.
\end{itemize}

This is because evidence monotonicity* applies to states pairwise, without
any consideration for the order of the pair. This has the interesting
implication that unlike measurability, there are no non-challengeable lies
under evidence monotonicity*.

\subsubsection{Main Result}

We now state the main result in this setting.

\begin{theorem}
Suppose that $\mathcal{E}_{i}\left( \cdot \right) $ is a costly evidence
structure and the designer can only make arbitrarily small rewards bounded
above by some $\varepsilon >0$. Then, an SCF $f$ is directly
Nash-implementable in mixed strategies in a direct revelation mechanism if
and only if it safisfies evidence monotonicity*.
\end{theorem}

The necessity component stems from the fact that if the designer is not
allowed to make large rewards, then in Definition 6, $t$ must be zero. If $%
t>0$, then the transfers cannot be bounded below any value of $\varepsilon $
which satisfies $\varepsilon <t$. Since evidence monotonicity is necessary
for implementation, this reinterpretation continues to be necessary when
only arbitrarily small rewards are allowed. We prove the above result by
constructing an implementing mechanism as follows.

\subsubsection{Implementing Mechanism}

\textbf{Message Space:} $M=\Pi _{i}M_{i}$, where $M_{i}=S\times \mathcal{E}%
_{i}$ (Direct Mechanism)

\textbf{Outcome:} $f(s_{1})$

\textbf{Transfers:}

The first transfer penalizes an agent $i$ an amount $T_{1}$ if other agents
do not present their cheapest evidence according to $i$'s state claim $s_{i}$%
. That is,

\begin{equation*}
\tau _{i}^{1}=-2NI\text{ if }\exists j\neq i\text{ s.t. }E_{j}\notin 
\mathcal{E}_{j}^{l}(s_{i})
\end{equation*}

The second transfer (which operates pairwise between agent $i$ and $j$
penalizes an agent $i$ an amount $|\mathcal{E}_{j}^{l}(s_{i})|$ if $%
s_{i}\neq s_{j}$. Under disagreement, this has the effect of incentivizing
agent $i$ to imply the smallest possible set of cheapest evidence for other
agents. Self-challengeable lies imply a set of cheapest evidence weakly
larger than that under the truth for other agents, but this transfer
incentivizes agents to reveal the true set of cheapest evidence for other
agents. Note that there is no conflict from this incentive on account of
different agents, since presenting the truth induces the smallest set of
cheapest evidence for all agents. Formally,

\begin{equation*}
\tau _{ij}^{2}=-2|\mathcal{E}_{j}^{l}(s_{i})|\text{ if }s_{i}\neq s_{j}
\end{equation*}

The third transfer penalizes an agent $i$ an amount $T_{3}$ if for some $j$,
agent $i$'s cheapest evidence is different under $s_{i}$ and $s_{j}$. That
is,%
\begin{equation*}
\tau _{i}^{3}=-1\text{ if }\exists j\neq i\text{ s.t. }\mathcal{E}%
_{i}^{l}(s_{i})\neq \mathcal{E}_{i}^{l}(s_{j})
\end{equation*}

Finally, the fourth transfer incentivizes an agent to challenge the lies of
other agents if possible. This incentive has to be carefully scaled, since
on the one hand, it must be positive to actually function as an incentive
for challenges, but not too large so as to either incentivize the
presentation of evidence from outside an agent's cheapest set of evidence or
result in transfers greater than $\varepsilon $. Accordingly, with $\gamma
\, $\ defined earlier as the minimum cost difference across all agents and
all states between cheapest and non-cheapest evidence, we define $\tau ^{4}$
as follows,

\begin{equation*}
\tau _{ij}^{4}=\frac{\min (\gamma ,\varepsilon )}{2I}\text{ if }E_{i}\notin 
\mathcal{E}_{i}^{l}(s_{j})\text{,}
\end{equation*}

where $I$ is the number of agents.

The total transfer agent $i$ is given by

\begin{equation*}
\tau ^{i}=\tau _{i}^{1}+\sum_{j\neq i}\tau _{ij}^{2}+\tau
_{i}^{3}+\sum_{j\neq i}\tau _{ij}^{4}
\end{equation*}

\subsubsection{Proof of Implementation}

Suppose the true state is $s^{\ast }$. Also, assume utilities are normalized
to $[0,1)$ so that an agent can be induced to accept any outcome if the
alternative is a penalty of $1$ unit. Consider any arbitrary mixed strategy
equilibrium $\sigma $.

\begin{claim}
\label{CHEAP}All agents present evidence from among their cheapest evidence
set, i.e. $E_{i}\in \mathcal{E}_{i}^{l}(s^{\ast })$
\end{claim}

\begin{proof}
By not presenting cheapest evidence, the maximum gain that $i$ might make
from $\tau ^{4}$ is $\frac{\gamma }{2}$, which is less than the difference
in cost he incurs from submitting an article of evidence which is not
amongst his cheapest. Thus, an agent always presents evidence from his
cheapest set of evidence.
\end{proof}

\begin{claim}
\label{NO_ORL}No agent presents an other-challengeable lie with positive
probability.
\end{claim}

\begin{proof}
Suppose instead that an agent $i$ presents a lie $s_{i}$ with positive
probability which an agent $j$ can challenge with the article of evidence $%
\overline{E_{j}}$.

First, we claim that agent $j$ must present, with a probability $\frac{1}{N}$
or more, an article of evidence $E_{j}^{\prime }$ (possibly equal to $%
\overline{E_{j}}$) which challenges some member of $RL_{j}$ which is being
presented with positive probability by some other agent.\ If not, there is a
message $m_{j}=(s_{j},E_{j})$ in the support of $\sigma _{i}$ such that $%
E_{j}$ does not challenge any member of $RL_{j}$. Then, on this message,
agent $j$ gains in expectation from $\tau ^{4}$ by presenting $\tilde{m}%
_{j}=(s_{j},\overline{E_{j}})$ instead of $m_{j}$. This gain is on account
of challenging agent $i$'s claim $s_{i}$. Since his evidence presentation
does not impact his payoff through the outcome or through any other
transfer, this is a profitable deviation.

Next, we claim that there is an agent $k\neq j$ (but possibly equal to $i$)
who plays a message $m_{k}=(s_{k},E_{k})$ with positive probability such
that $E_{j}^{\prime }$ challenges $s_{k}$. This is because we have
established that $E_{j}^{\prime }$ is played with a probability at least $%
\frac{1}{N}$, and if there is no such $s_{k}$ which is challenged by $%
E_{j}^{\prime }$, then $j$ could profitably deviate by replacing $%
E_{j}^{\prime }$ with $\overline{E_{j}}$ in any message where he presents $%
E_{j}^{\prime }$.

Finally, we conclude the proof by claiming that agent $k$ can profitably
deviate by presenting $\tilde{m}_{k}=(s^{\ast },E_{k})$ instead of $m_{k}$.
This is because when presenting $m_{k}$, agent $k$ expects that agent $j$
presents $E_{j}^{\prime }$ with a probability at least $\frac{1}{N}$, and
thus with probability $\frac{1}{N}$, he is fined $2NI$ owing to $\tau ^{1}$.
Since Claim \ref{CHEAP} establishes that agents only present their cheapest
evidence, he can avoid this fine by presenting $s^{\ast }$. The table below
shows his change in payoff from presenting $\tilde{m}_{k}$ instead of $m_{k}$%
.%
\begin{equation*}
\begin{tabular}{|l|l|l|l|l|l|}
\hline
$g$ & $\tau _{i}^{1}$ & $\tau _{i}^{2}$ & $\tau _{i}^{3}$ & $\tau _{i}^{4}$
& In total \\ \hline
$>-1$ & $\geq \frac{1}{N}\times 2NI$ & $\geq -2(I-1)$ & $\geq -1$ & $0$ & $%
>0 $ \\ \hline
\end{tabular}%
\end{equation*}

He loses less than $1$ to the outcome, $2(I-1)$ to $\tau ^{2}$, and $1$ to $%
\tau ^{3}$ at worst. Since the evidence presentation is unaltered, there is
no loss from $\tau ^{4}$. This is, therefore a profitable deviation.
\end{proof}

In view of the above claim, agents are now restricted to presenting only the
truth or self-challengeable lies with positive probability.

\begin{claim}
\label{SRL_1}Agents only present such state claims with positive probability
which induce the true set of cheapest evidence for other agents. That is,
for any $s_{i}$ that an agent $i$ presents with positive probability, $%
\mathcal{E}_{j}^{l}(s_{i})=\mathcal{E}_{j}^{l}(s^{\ast })$ for all $j\neq i$.
\end{claim}

\begin{proof}
Suppose not. Then, agent $i$ presents a message $m_{i}=(s_{i},E_{i})$ such
that there is an agent $j\neq i$ for whom $\mathcal{E}_{j}^{l}(s_{i})\not=%
\mathcal{E}_{j}^{l}(s^{\ast })$. We note that $s_{i}$ must then be a
self-challengeable lie.\footnote{%
Since there are no non-challengeable lies under this setup, the only states
are the truth, self-challengeable lies, and other-challengeable lies.} Now,
consider a deviation to $\tilde{m}_{i}=(s^{\ast },E_{i})$ instead of $m_{i}$%
. The table below shows his change in payoff from presenting $\tilde{m}_{i}$
instead of $m_{k}$.%
\begin{equation*}
\begin{tabular}{|l|l|l|l|l|l|}
\hline
$g$ & $\tau _{i}^{1}$ & $\tau _{i}^{2}$ & $\tau _{i}^{3}$ & $\tau _{i}^{4}$
& In total \\ \hline
$>-1$ & $\geq 0$ & $\geq 2$ & $\geq -1$ & $0$ & $>0$ \\ \hline
\end{tabular}%
\end{equation*}

First, the agent incurs no losses from $\tau ^{1}$, since only members of
the cheapest set of evidence are presented and they are consistent with the
truth. Second, a self-challengeable lie cannot be identical to any claim
that is presented with positive probability by other agents, so that when
presenting $m_{i}$, agent $i$ expects $\tau ^{2}$ to be active with
probability 1. From Observation 1*, if $\mathcal{E}_{j}^{l}(s_{i})\not=%
\mathcal{E}_{j}^{l}(s^{\ast })$, then $\mathcal{E}_{j}^{l}(s^{\ast })\subset 
\mathcal{E}_{j}^{l}(s_{i})$. Then, the agent gains at least $2$ from this
deviation owing to $\tau ^{2}$. Third, the agent may lose at most $1$ from
this deviation owing to $\tau ^{3}$, since other agents may not be playing
states that induce the true set of cheapest evidence for $i$. Fourth, since
the evidence presentation is unaltered, there is no loss from $\tau ^{4}$.
Finally, the agent may lose at most 1 to the outcome. Overall, as shown
above, this is a profitable deviation.
\end{proof}

\begin{claim}
\label{SRL_2}No agent presents a self-challengeable lie.
\end{claim}

\begin{proof}
Suppose instead that an agent $i$ presents a message $m_{i}=(s_{i},E_{i})$
such that $s_{i}$ is self-challengeable. Then, $\mathcal{E}%
_{i}^{l}(s_{i})\not=\mathcal{E}_{i}^{l}(s^{\ast })$. From Claim \ref{SRL_1},
other agents only present state claims which are consistent with agent $i$'s
true set of cheapest evidence. Then, on $m_{i}$, the agent $i$ expects a
penalty of $1$ owing to $\tau ^{3}$. Consider a deviation to the truth. That
is, consider a deviation to $\tilde{m}_{i}=(s^{\ast },E_{i})$ instead of $%
m_{i}$. The table below shows his change in payoff from presenting $\tilde{m}%
_{k}$ instead of $m_{k}$.%
\begin{equation*}
\begin{tabular}{|l|l|l|l|l|l|}
\hline
$g$ & $\tau _{i}^{1}$ & $\tau _{i}^{2}$ & $\tau _{i}^{3}$ & $\tau _{i}^{4}$
& In total \\ \hline
$>-1$ & $\geq 0$ & $\geq 0$ & $\geq 1$ & $0$ & $>0$ \\ \hline
\end{tabular}%
\end{equation*}

As discussed earlier, the agent incurs no loss from $\tau ^{1}$. Claim \ref%
{SRL_1} implies that $s_{i}$ and $s^{\ast }$ are identical in terms of the
cheapest evidence set of any other agent, so that this deviation does not
impact the payoff from $\tau ^{2}$. The agent gains $1$ from $\tau ^{3}$
since the truth is consistent with the cheapest evidence set being implied
by the other agents, whereas $s_{i}$ was not. Further, the payoff from $\tau
^{4}$ is unaffected because the evidence presentation remains the same.
Finally, the agent may lose $1$ to the outcome. Overall, this is a
profitable deviation.
\end{proof}

\begin{claim}
The mechanism implements and there are no transfers in equilibrium.
\end{claim}

\begin{proof}
From the earlier claims, agents only present the truth and cheapest evidence
in equilibrium. Then, it is clear that no challenges occur, so that $\tau
^{1}$ and $\tau ^{4}$ are inactive. Since all the agents present the truth, $%
\tau ^{2}$ and $\tau ^{3}$ are also inactive. Therefore, the correct outcome
is obtained alongside no transfers.
\end{proof}

\subsection{Implementation without Normality}

For any state\ $s$, denote by $T^{f}(s)=\{s^{\prime }:(s,s^{\prime })$
violates Maskin Monotonicity$\}$ the set of states which prescribe different
optimal outcomes relative to $s$. Notice that without preference variation,
every pair of states which prescribe different optimal outcomes violates
Maskin monotonicity. In what follows, we use the same implementation notion
as that in Theorem 1 of the main paper (i.e., Definition 3 of \cite{BCS2020}%
).

We now state the main result in this setting.

\begin{theorem}
An SCF is implementable if and only if for each state $s$, $s$ and $T^{f}(s)$
are distinguishable.
\end{theorem}

Necessity follows from Proposition 3 in \cite{KT2012}. Sufficiency is
established in the following mechanism.

\subsubsection{Implementing Mechanism}

In what follows, define $s\thickapprox s^{\prime }$ if $f(s)=f(s^{\prime })$%
, and $s\not\thickapprox s^{\prime }$ otherwise, and suppose that $N$ is the
maximum cardinality of any evidence set.

\paragraph{\textbf{Message Space}}

Every agent has a typical message $m_{i}=(s_{i},E_{i}^{1},E_{i}^{2})\in
M_{i}=S\times \mathcal{E}_{i}\times \mathcal{E}_{i}\cup \phi $.

\paragraph{\textbf{Outcome}}

We define the outcome function of the mechanism as%
\begin{equation*}
g\left( m\right) =f\left( s_{1}\right) \text{.}
\end{equation*}%
That is, we implement the social outcome according to the state claim made
by the first agent. However, we will show that in any equilibrium all agents
report the true state. Hence, any Nash equilibrium achieves the desirable
social outcome.

\paragraph{\textbf{Transfers}}

The first transfer provides agents an incentive to refute such lies from
other agents as they can, and also to avoid refutations of their own state
claims if possible.

\begin{equation*}
\tau _{ij}^{1}\left( m\right) =\left\{ 
\begin{tabular}{ll}
$T_{1}$, & if $s_{i}\in E_{j}^{1}$ and $s_{j}\not\in E_{i}^{1}$; \\ 
$-T_{1}$, & if $s_{i}\not\in E_{j}^{1}$ and $s_{j}\in E_{i}^{1}$; \\ 
$0$, & otherwise.%
\end{tabular}%
\right.
\end{equation*}

Here, $T_{1}=2N+4N^{2}I+2N^{3}I$.

Define $s\succ s^{\prime }$ iff $\forall i$,$\mathcal{E}_{i}(s^{\prime
})\subseteq \mathcal{E}_{i}(s)$ and there exists $j$ such that, $\mathcal{E}%
_{j}(s^{\prime })\subset \mathcal{E}_{j}(s)$. Otherwise, $s\not\succ
s^{\prime }$. For any $s$, denote by $\hat{E}(s)$ the article of evidence
(if it exists) which refutes each $s^{\prime }$ such that $s\succ s^{\prime
} $. If $\NEG{\exists}s^{\prime }$ such that $s\succ s^{\prime }$, then $%
\hat{E}(s)=\phi $.

Define $C_{i}(m)=1$ if for any $j\in \mathcal{I}$, $E_{j}^{2}\neq \hat{E}%
(s_{i})$ and $E_{j}^{2}\neq \hat{E}(s)$ for some $s\thickapprox s_{i}$, and $%
C_{i}(m)=0$ otherwise. This scenario with $C_{i}(m)=1$ is termed "no
support", as $\hat{E}(s_{i})$ is interpreted as supporting the state claim $%
s_{i}$ by refuting states which occur below $s_{i}$ in the order defined
above. Note that it is not necessary to actually support the exact state $%
s_{i}$, so that it suffices to support some state $s$ which is identical
under the $\thickapprox $ relation. The transfer below incentivizes agents
to attempt to imply the smallest possible set of evidence for every agent in
the state claim they make. This helps with eliminating nonrefutable lies
since the truth provides a strictly lower loss from this transfer.

\begin{equation*}
\tau _{i}^{2}\left( m\right) =\left\{ 
\begin{tabular}{ll}
$-2\sum_{k\in \mathcal{I}}|\mathcal{E}_{k}(s_{i})|$, & if $C_{i}(m)=1$; \\ 
$0$, & otherwise.%
\end{tabular}%
\right. \text{.}
\end{equation*}

Further, define $D(m)=1$ if $\exists (l,m)\in \mathcal{I}\times \mathcal{I}$
such that $s_{l}\not\thickapprox s_{m}$, and $D(m)=0$ otherwise. The
following two transfers help to eliminate self refutable lies using a
crosscheck similar in spirit to that in the proof of Theorem 1.

\begin{equation*}
\tau _{ij}^{3}=-T_{3}|\mathcal{E}_{j}(s_{i})|\text{ if }D(m)=1\text{ or }%
C(m)=1\text{,}
\end{equation*}

where $T_{3}=2+2N$. Further,

\begin{equation*}
\tau _{i}^{4}=-1\text{ if }\exists j\neq i\text{ such that }\mathcal{E}%
_{i}(s_{i})\neq \mathcal{E}_{i}(s_{j})\text{.}
\end{equation*}

With $\tau _{i}^{1}=\sum_{j\neq i}\tau _{ij}^{1}$ and $\tau
_{i}^{3}=\sum_{j\neq i}\tau _{ij}^{3}$, we define the overall transfer to
agent $i$ as:

\begin{equation*}
\tau ^{i}=\tau _{i}^{1}+\tau _{i}^{2}+\tau _{i}^{3}+\tau _{i}^{4}\text{.}
\end{equation*}

\paragraph{Proof of Implementation}

\begin{claim}
\label{NOORL}No agent presents an other-refutable lie with positive
probability.
\end{claim}

\begin{proof}
Suppose instead that an agent $i$ presents a lie $s_{i}$ with positive
probability which an agent $j$ can challenge with the article of evidence $%
\bar{E}_{j}$.

First, we claim that agent $j$ must present, with a probability $\frac{1}{N}$
or more, an article of evidence $E_{j}^{\prime }$ (possibly equal to $\bar{E}%
_{j}$) which refutes some member of $RL_{j}$ which is being presented with
positive probability by some other agent.\ If not, there is a message $%
m_{j}=(s_{j},E_{j}^{1},E_{j}^{2})$ in the support of $\sigma _{j}$ such that 
$E_{j}^{1}$ does not refute any member of $RL_{j}$. Then, on this message,
agent $j$ gains in expectation from $\tau ^{1}$ by presenting $\tilde{m}%
_{j}=(s_{j},\bar{E}_{j},E_{j}^{2})$ instead of $m_{j}$. This gain is on
account of refuting agent $i$'s claim $s_{i}$. Since his evidence
presentation in $E^{1}$ does not impact his payoff through the outcome or
through any other transfer, this is a profitable deviation.

Next, we claim that there is an agent $k\neq j$ (but possibly equal to $i$)
who plays a message $m_{k}=(s_{k},E_{k}^{1},E_{k}^{2})$ with positive
probability such that $E_{j}^{\prime }$ refutes $s_{k}$. This is because we
have established that $E_{j}^{\prime }$ is played with a probability at
least $\frac{1}{N}$, and if there is no such $s_{k}$ which is refuted by $%
E_{j}^{\prime }$, then $j$ could profitably deviate by replacing $%
E_{j}^{\prime }$ with $\bar{E}_{j}$ in any message where he presents $%
E_{j}^{\prime }$.

Finally, we conclude the proof by claiming that agent $k$ can profitably
deviate by presenting $\tilde{m}_{k}=(s^{\ast },E_{k}^{1},E_{k}^{1})$
instead of $m_{k}$. This is because when presenting $m_{k}$, agent $k$
expects that agent $j$ presents $E_{j}^{\prime }$ with a probability at
least $\frac{1}{N}$, and thus with probability $\frac{1}{N}$, he is fined $%
T_{1}$ owing to $\tau ^{1}$. The table below shows his change in payoff from
presenting $\tilde{m}_{k}$ instead of $m_{k}$.%
\begin{equation*}
\begin{tabular}{|l|l|l|l|l|l|}
\hline
$g$ & $\tau _{i}^{1}$ & $\tau _{i}^{2}$ & $\tau _{i}^{3}$ & $\tau _{i}^{4}$
& In total \\ \hline
$>-1$ & $\geq \frac{1}{N}\times T_{1}$ & $\geq -2NI$ & $\geq -T_{3}NI$ & $%
\geq -1$ & $>0$ \\ \hline
\end{tabular}%
\end{equation*}

Therefore, this is a profitable deviation.
\end{proof}

\begin{claim}
No agent $i$ presents a nonrefutable lie $s_{i}\not\thickapprox s^{\ast }$
with positive probability.
\end{claim}

\begin{proof}
Suppose to the contrary that an agent presents a message $%
m_{i}=(s_{i},E_{i}^{1},E_{i}^{2})$ with positive probability so that $s_{i}$
is a nonrefutable lie at $s^{\ast }$ and $s_{i}\not\thickapprox s^{\ast }$.
Since $\hat{E}(s_{i})$ (which exists due to the fact that $s_{i}$ must be
distinguishable from $s^{\ast }$) refutes $s^{\ast }$, it cannot be
available at $s^{\ast }$ and therefore $\tau _{i}^{2}$ is active on the
message $m_{i}$. Consider instead a deviation to the truth, i.e. presenting $%
s^{\ast }$ instead of $s_{i}$. The following table shows the change in his
payoff from this change.

\begin{equation*}
\begin{tabular}{|l|l|l|l|l|l|}
\hline
$g$ & $\tau _{i}^{1}$ & $\tau _{i}^{2}$ & $\tau _{i}^{3}$ & $\tau _{i}^{4}$
& In total \\ \hline
$>-1$ & $0$ & $\geq 2$ & $\geq 0$ & $\geq -1$ & $>0$ \\ \hline
\end{tabular}%
\end{equation*}

The agent makes no losses from $\tau ^{1}$ since the truth is irrefutable.
Further, the agent makes no losses from $\tau ^{3}$ because while the agent
endowed with $\hat{E}(s^{\ast })$ (if it exists) may not be presenting it, $%
\hat{E}(s_{i})$ was not presented either, so that this transfer were active
under $m_{i}$ as well and $|\mathcal{E}_{j}(s_{i})|\geq |\mathcal{E}%
_{j}(s^{\ast })|$ for all agents $j\neq i$. Therefore, this is a profitable
deviation.
\end{proof}

\begin{claim}
If any agent $i$ presents a self-refutable lie $s_{i}$, it must be
consistent with the truth along other agents' evidence, i.e. $E_{j}^{\ast
}(s_{i})=E_{j}^{\ast }(s^{\ast })\forall j\neq i$.
\end{claim}

\begin{proof}
Suppose for contradiction that agent $i$ presents a self-refutable lie $%
s_{i} $ such that there exists an agent $j$ for whom $E_{j}^{\ast
}(s_{i})\not=E_{j}^{\ast }(s^{\ast })$. Any self-refutable lie that is
presented must disagree with the reports of others (since others are not
presenting other-refutable lies) therefore holding $\tau ^{3}$ active.
Consider a deviation to the truth, i.e. presenting $s^{\ast }$ instead of $%
s_{i}$. The following table shows the change in his payoff from this change.%
\begin{equation*}
\begin{tabular}{|l|l|l|l|l|l|}
\hline
$g$ & $\tau _{i}^{1}$ & $\tau _{i}^{2}$ & $\tau _{i}^{3}$ & $\tau _{i}^{4}$
& In total \\ \hline
$>-1$ & $0$ & $\geq -2N$ & $\geq T_{3}$ & $\geq -1$ & $>0$ \\ \hline
\end{tabular}%
\end{equation*}

The gain from $\tau ^{3}$ emerges from the fact that $|\mathcal{E}%
_{j}(s_{i})|>|\mathcal{E}_{j}(s^{\ast })|$ if $s_{i}$ is self-refutable for $%
i$. Agent $i$ loses at most $2N$ from $\tau ^{2}$ as switching from a
self-refutable lie to the truth weakly reduces the size of the implied
evidence set for other agents, but may increase the size of the evidence set
for agent $i$ himself. SInce, $T_{3}=2+2N$, this is a profitable deviation.
\end{proof}

Denote by $\hat{\imath}$ the agent who is endowed with $\hat{E}(s^{\ast })$
(if it exists). Otherwise, $\hat{\imath}=1$.

\begin{claim}
Agent $\hat{\imath}$ does not present a self-refutable lie.
\end{claim}

\begin{proof}
Suppose to the contrary that agent $\hat{\imath}$ presents with positive
probability a message $m_{\hat{\imath}}=(s_{\hat{\imath}},E_{\hat{\imath}%
}^{1},E_{\hat{\imath}}^{2})$ such that $s_{\hat{\imath}}$ is a
self-refutable lie. Consider a deviation to an alternate message $\tilde{m}_{%
\hat{\imath}}=(s^{\ast },E_{\hat{\imath}}^{1},\hat{E}(s^{\ast }))$. The
following table shows the change in his payoff from this change.%
\begin{equation*}
\begin{tabular}{|l|l|l|l|l|l|}
\hline
$g$ & $\tau _{i}^{1}$ & $\tau _{i}^{2}$ & $\tau _{i}^{3}$ & $\tau _{i}^{4}$
& In total \\ \hline
$>-1$ & $0$ & $\geq 0$ & $\geq 0$ & $1$ & $>0$ \\ \hline
\end{tabular}%
\end{equation*}

The agent suffers no losses from $\tau ^{2}$ since his report of the truth
is now supported by the accompanying evidence $\hat{E}(s^{\ast })$, and
suffers no losses from $\tau ^{3}$ since $\tau ^{3}$ was active on $m_{\hat{%
\imath}}$ (since $D(m)=1$ on $m_{i}$) and $|\mathcal{E}_{j}(s_{i})|\geq |%
\mathcal{E}_{j}(s^{\ast })|$ if $s_{i}$ is self-refutable for $i$. The agent
gains from $\tau ^{4}$ since $\mathcal{E}_{\hat{\imath}}(s_{\hat{\imath}%
})\neq \mathcal{E}_{\hat{\imath}}(s_{j})$ but $\mathcal{E}_{\hat{\imath}%
}(s^{\ast })=\mathcal{E}_{\hat{\imath}}(s_{j})$ from the previous claim.
Therefore, this is a profitable deviation.
\end{proof}

\begin{claim}
No agent presents self-refutable lies with positive probability.
\end{claim}

\begin{proof}
From the above claim, agent $\hat{\imath}$ presents the truth. $\hat{E}%
(s^{\ast })$, we claim, must be submitted in equilibrium alongside the state
report $s^{\ast }$ by some agent, since if it is not, then it is a
profitable deviation for agent $\hat{\imath}$ to do so. Therefore, in any
equilibrium, the truth is supported.

Now, suppose for contradiction that some agent $i$ presents with positive
probability a message $m_{i}=(s_{i},E_{i}^{1},E_{i}^{2})$ such that $s_{i}$
is a self-refutable lie. Consider a deviation to an alternate message $%
\tilde{m}_{\hat{\imath}}=(s^{\ast },E_{i}^{1},E_{i}^{2})$. The following
table shows the change in his payoff from this change.%
\begin{equation*}
\begin{tabular}{|l|l|l|l|l|l|}
\hline
$g$ & $\tau _{i}^{1}$ & $\tau _{i}^{2}$ & $\tau _{i}^{3}$ & $\tau _{i}^{4}$
& In total \\ \hline
$>-1$ & $0$ & $\geq 0$ & $\geq 0$ & $1$ & $>0$ \\ \hline
\end{tabular}%
\end{equation*}

As earlier, the agent faces no losses from $\tau ^{2}$ and $\tau ^{3}$ as
the self-refutable lie guarantees disagreement, while the truth may yield
agreement and is supported. The agent gains from $\tau ^{4}$ since each
agent implies the same evidence as that implied by the truth for all other
agents, and $\mathcal{E}_{\hat{\imath}}(s_{\hat{\imath}})\neq \mathcal{E}_{%
\hat{\imath}}(s_{j})$ but $\forall j$, $\mathcal{E}_{\hat{\imath}}(s^{\ast
})=\mathcal{E}_{\hat{\imath}}(s_{j})$ for any $s_{j}$ that is presented with
positive probability. Therefore, this is a profitable deviation.
\end{proof}

\begin{claim}
The outcome is $f(s^{\ast })$ and there are no transfers in equilibrium.
\end{claim}

\begin{proof}
By the claims above, all agents present the truth (i.e. $s^{\ast }$) with
probability 1 so that the outcome is $f(s^{\ast })$.This joint report is not
refutable, so that $\tau ^{1}$ yields no transfers. Further, the truth is
supported, so that $\tau ^{2}$ and $\tau ^{3}$ yield no transfers. Since
there is a joint report of $s^{\ast }$, the transfer $\tau ^{4}$ is also
inactive. Therefore, we have no transfers.
\end{proof}

\begin{remark}
Without nonrefutable lies, the mechanism can be simplified to a direct
mechanism, as we would no longer need $E_{\symbol{126}}^{2}$. This is
because of the following chain of logic. First, $\tau ^{1}$ eliminates
other-refutable lies. Then, since other-refutable lies have been eliminated,
any self-refutable lie produces disagreements with probability one, so that
the transfers $\tau ^{3}$ and $\tau ^{4}$ are active. In turn, $\tau ^{3}$
yields that agents only present self-refutable lies which induce the same
evidence for others as the truth. Finally, $\tau ^{4}$ yields that deviating
away from a self-refutable lie to the truth is a profitable deviation.
\end{remark}

\subsection{Renegotiation-Proof Contracting}

\subsubsection{\label{app:RP}Proof of Theorem \protect\ref{RPT_REINTERPRET}}

We begin by proving that it is necessary for the existence of
renegotiation-proof contracts that either Condition \emph{(a)} or (\emph{b)}
be satisfied. For a contradiction, consider a social choice function $f$ and
suppose that both Condition \emph{(a)} and Condition \emph{(b)} are violated
for $f$. Then, there exist a pair of states $s$ and $s^{\prime }$ such that $%
f(s)\neq f(s^{\prime })$ and one of the following is true:

\begin{enumerate}
\item[(c)] One of the agents can refute $s^{\prime }$ at $s$ and the other
agent can refute $s$ at $s^{\prime }$; or

\item[(d)] $s^{\prime }$ is nonrefutable at $s$, but only one agent can
refute $s$ at $s^{\prime }$.
\end{enumerate}

Further, suppose that there is a mechanism $\mathcal{M}=(M,g,\left( \tau
_{i}\right) _{i\in \mathcal{I}})$ which implements for every pair of $v$ and 
$h$ in Nash Equilibrium with renegotiation. We note that any such mechanism
must be such that $\sum_{i}\tau _{i}=0$. Indeed, if $\mathcal{M}$ results in
a budget surplus ($\sum_{i}\tau _{i}<0$), then such an allocation is not
efficient with respect to $v$ and will get renegotiated under $h$. Consider
any equilibrium $\sigma $ of this mechanism and further consider any pair of
messages $\left( m_{i}^{s},m_{j}^{s}\right) $ on the support of $\sigma $ at 
$s$ and another pair of messages $\left( m_{i}^{s^{\prime
}},m_{j}^{s^{\prime }}\right) $ on the support of $\sigma $ at $s^{\prime }$%
. Further, assume $v_{i}(f(s^{\prime }))=1$, $v_{i}(f(s))=0$, $%
v_{j}(f(s^{\prime }))=0$, $v_{j}(f(s))=1$, and for all other members $a$ of
the set $A$, $v_{i}(a)+v_{j}(a)=1$, $v_{i}(a)<1$, $v_{j}(a)<1$. Here, $v$ is
state independent. Further, choose any $h$ such that $v$ and $h$ satisfy the
constraints noted earlier.\footnote{%
Since $\mathcal{M}$ must implement for every acceptable combination of $v,h$
and $f$, we take the liberty of choosing $v$. Moreover, our argument works
for any $h$ which satisfies the constraints noted above; e.g., $h$ may be
set as the identity mapping which corresponds to \textquotedblleft
no-renegotiation\textquotedblright .} To be concise, we will use $\bar{u}%
_{i}(m)$ to denote the utility agent $i$ derives from the outcome and
transfers chosen by $\mathcal{M}$ under the message profile $m$.

If Condition \emph{(c)} is satisfied, then without loss of generality, the
only possible evidence structure is that agent $i$ can refute $s^{\prime }$
at $s$ and agent $j(\neq i)$ can refute $s$ at $s^{\prime }$. Since $%
\mathcal{M}$ has $\sum_{i}\tau _{i}=0$, $\bar{u}_{i}(m_{i}^{s^{\prime
}},m_{j}^{s})+\bar{u}_{i}(m_{i}^{s^{\prime }},m_{j}^{s})=1$. We consider the
following two cases.

\noindent \textbf{Case 1: }If the mechanism $\mathcal{M}$ is such that $\bar{%
u}_{i}(m_{i}^{s^{\prime }},m_{j}^{s})>\bar{u}%
_{i}(m_{i}^{s},m_{j}^{s})=v_{i}(f(s))$, then agent $i$ deviates to $%
m_{i}^{s^{\prime }}$ in state $s$ and this yields an outcome different from $%
f(s)$, so that it is not possible that $\mathcal{M}$ implements $f$ in Nash
equilibrium with renegotiation. Note that this is feasible for agent $i$.

\noindent \textbf{Case 2:} If, on the other hand, $\mathcal{M}$ is such that 
$\bar{u}_{i}(m_{i}^{s^{\prime }},m_{j}^{s})\leq v_{i}(f(s))$, then $\bar{u}%
_{j}(m_{i}^{s^{\prime }},m_{j}^{s})\geq v_{j}(f(s))>v_{j}(f(s^{\prime }))$
and agent $j$ deviates to $m_{j}^{s}$ in state $s^{\prime }$.

\noindent In either case, $\mathcal{M}$ cannot implement $f$ in Nash
equilibrium with renegotiation.

If Condition (\emph{d})\emph{\ }is satisfied, suppose that $s^{\prime }$ is
nonrefutable at $s$ and without loss of generality, that only agent $j$ can
refute $s$ at $s^{\prime }$. Since $\mathcal{M}$ has $\sum_{i}\tau _{i}=0$, $%
\bar{u}_{i}(m_{i}^{s^{\prime }},m_{j}^{s})+\bar{u}_{i}(m_{i}^{s^{\prime
}},m_{j}^{s})=1$. We consider the following two cases.

\noindent \textbf{Case 1: }If the mechanism $\mathcal{M}$ is such that $\bar{%
u}_{j}(m_{i}^{s^{\prime }},m_{j}^{s})>v_{j}(f(s^{\prime }))$, then agent $j$
deviates to $m_{j}^{s}$ in state $s^{\prime }$ and this yields an outcome
different from $f(s^{\prime })$.

\noindent \textbf{Case 2:} If, on the other hand, $\mathcal{M}$ is such that 
$\bar{u}_{j}(m_{i}^{s^{\prime }},m_{j}^{s})\leq v_{j}(f(s^{\prime }))$, then 
$\bar{u}_{i}(m_{i}^{s^{\prime }},m_{j}^{s})\geq v_{i}(f(s^{\prime
}))>v_{i}(f(s))$ and agent $i$ deviates to $m_{i}^{s^{\prime }}$ in state $s$%
.

\noindent In either case, $\mathcal{M}$ cannot implement $f$ in Nash
equilibrium with renegotiation.

To prove that either of Conditions \emph{(a)} and \emph{(b) }are sufficient,
we present the following mechanism. The message space is the same as our
main mechanism: $M_{i}=S\times \mathcal{E}_{i}$. The outcome is given by $%
f(s_{1})$, where $f$ is the SCF we desire to implement. There are two
transfers: A fine of 1 dollar for agent $i$ if he is unable to support his
state claim and a fine of 2 dollars if $s_{i}$ is refuted by $E_{j}$. Each
fine is paid to the other agent, so that the mechanism is budget balanced.
Also recall that the maximum bound of utility is less than 1 dollar, so that
these transfers dominate the utility from the outcome.

Assume that the true state is $s^{\ast }$. We now prove that this mechanism
implements $f$ with renegotiation.

First, no agent presents an other-refutable lie with positive probability;
otherwise, since $\mathcal{E}$ is normal, the other agent presents the
refuting evidence with probability 1 and this loses the first agent 2
dollars of money, which is more than the value of influencing the outcome.
Second, if an agent presents a self-refutable lie $s$ with positive
probability, then Condition \emph{(a)\ }yields that this agent can also
refute $s^{\ast }$ at $s$. Therefore, the agent cannot support his state
claim, and is fined 1 dollar. As this is more than the value of the outcome,
deviating to $(s^{\ast },E_{i}^{\ast }(s^{\ast }))$ is profitable. Third, if
an agent presents a nonrefutable lie $s$, then Condition \emph{(b)} yields
that $i$ can refute $s^{\ast }$ at $s$. Thus, he cannot support his state
claim, and again, deviating to $(s^{\ast },E_{i}^{\ast }(s^{\ast }))$ is
profitable. Therefore, the only equilibrium is truth-telling, which results
in the outcome $f(s)$ without transfers, which is efficient.

We note here that all the off-equilibrium outcomes above involved $f(s_{1})$
with a budget balanced transfer of value greater than 1 dollar, so that the
resulting outcomes are Pareto efficient. Therefore, the mechanism is
invulnerable to renegotiation.

\subsubsection{\label{app:RPEvidenceEx} Evidence Structure which does not
allow for renegotiation proof contracting}

\begin{equation*}
\begin{tabular}{|c|c|c|}
\hline
Agent/State & $\phi $ & $\theta $ \\ \hline
Buyer & $\{\{\theta ,\phi \}\}$ & $\{\{\theta \},\{\theta ,\phi \}\}$ \\ 
\hline
Seller & $\{\{\theta ,\phi \}\}$ & $\{\{\theta ,\phi \}\}$ \\ \hline
\end{tabular}%
\end{equation*}

This evidence structure does not allow for renegotiation proof contracting
since at state $\phi ,$ state $\theta $ is a nonrefutable lie, but at state $%
\theta ,$ only the buyer can refute the state $\phi $. In contrast, if in
the state with low investment, $\phi $, the buyer could refute $\theta $,
then renegotiation-proof contracting would once again be possible.

\bibliographystyle{econometrica}
\bibliography{acompat,dnie}

\end{document}